\documentclass[preprint,12pt]{article}
\usepackage[top=1.1in, bottom=1.2in, left=0.8in, right=0.8in]{geometry}
\usepackage{amsmath,amsfonts,amssymb,amsthm,mathtools}
\usepackage{mathrsfs}
\usepackage{graphicx}
\usepackage{xcolor}
\usepackage{colortbl,dcolumn}
\usepackage{psfrag}
\usepackage{booktabs,multirow,array,longtable}
\usepackage{cite}
\usepackage{url}
\usepackage{hyperref}
\hypersetup{
  colorlinks=false,
  pdfborder={0 0 0.7},
  linkbordercolor={1 0 0},
  citebordercolor={0 0.75 0},
  urlbordercolor={0 0 1}
}
\usepackage{enumitem}
\usepackage{caption}
\usepackage{subcaption}
\usepackage{float}
\usepackage[ruled,vlined,linesnumbered]{algorithm2e}
\usepackage{microtype}
\allowdisplaybreaks[4]
\numberwithin{equation}{section}

\newcommand{\R}{\mathbb{R}}

\newcommand{\E}{\mathbb{E}}

\newcommand{\dd}{\mathrm{d}}

\newcommand{\norm}[1]{\left\lVert #1\right\rVert}

\newcommand{\ip}[2]{\left\langle #1,#2\right\rangle}

\newtheorem{theorem}{Theorem}[section]
\newtheorem{definition}[theorem]{Definition}
\newtheorem{lemma}[theorem]{Lemma}
\newtheorem{proposition}[theorem]{Proposition}
\newtheorem{corollary}[theorem]{Corollary}
\newtheorem{remark}[theorem]{Remark}

\newtheorem{assumption}{Assumption}

\setlist[itemize]{leftmargin=2em}
\setlist[enumerate]{leftmargin=2em}
\DontPrintSemicolon
\SetKwInput{KwInput}{Input}
\SetKwInput{KwOutput}{Output}
\SetKw{KwReturn}{return}
\SetAlgoNlRelativeSize{-1}

\begin{document}
\title{Deterministic Denominator Design for Localized Tamed Stochastic Gradient Langevin Dynamics}

\author{
Yiwei Zhou\thanks{School of Mathematics and Statistics, Yunnan University, Kunming, Yunnan 650500, China. Email: \texttt{yiwei.zhou@utexas.edu}.}
\and
Ziheng Chen\thanks{School of Mathematics and Statistics, Yunnan University, Kunming, Yunnan 650500, China. Email: \texttt{12024113103@stu.ynu.edu.cn}.}
}

\date{}
\maketitle

\begin{abstract}
If the denominator in a tamed stochastic gradient Langevin update uses the current stochastic-gradient draw, the conditional mean can be biased even when the stochastic-gradient oracle is unbiased.  A state-dependent denominator fixed before that draw removes this coupling.  We build practical deterministic denominators from a short pilot run.  A log-scale proxy is fitted to the growth score \(G_\star(x)=\|b(x)\|/(1+\|x\|)\), and empirical pilot quantiles set the activation thresholds of a local proxy-quantile envelope.  The reported proxy-quantile experiments use this local denominator.  Separately, we describe a final denominator with a norm-polynomial tail-floor correction that can be used when one wants to certify the global effective-linearity input required by the companion deterministic-envelope Lyapunov theory.  We show how proxy and threshold errors enter denominator errors and the resulting stationary observable errors.  In the reported experiments, the local proxy-quantile denominator improves over random-denominator tamed SGLD at comparable production cost.  It also gives observable behavior close to the $G_\star$-envelope benchmark, without the full-gradient growth-score evaluations required by that benchmark.


\end{abstract}

\section{Introduction}
Stochastic gradient Langevin dynamics (SGLD) replaces the full gradient in a Langevin update by a stochastic-gradient oracle.  For nonconvex objectives \(b(x)\) with superlinear growth, fixed-step explicit updates may be unstable unless the drift is tamed.  A denominator can stabilize large increments, but it also changes the Markov transition.  The question is therefore not only whether to tame, but how to choose the denominator.

The companion localized-taming paper~\cite{companion} gives the first design rule.  If the denominator is evaluated from the current stochastic-gradient draw, write its envelope term as \(\widetilde A_\eta(x,U)\).  Then generally
\begin{equation*}
  \E_U\left[
  \frac{\widehat g_m(x,U)}{1+\eta^\alpha \widetilde A_\eta(x,U)}
  \mid x\right]
  \neq
  \frac{\nabla F_z(x)}{1+\eta^\alpha A_\eta(x)},
\end{equation*}
even when \(\widehat g_m\) is unbiased.  This mean-shift channel is removed when the denominator is fixed once the current state is known:
\begin{equation*}
  X_{k+1}=X_k-
  \eta\,\frac{\widehat g_m(X_k,U_k)}{1+\eta^\alpha A_\eta(X_k)}
  +\sqrt{2\eta/\beta}\,Z_{k+1}.
\end{equation*}
This paper starts from that deterministic-denominator principle and uses the localized envelope from the companion framework:
\begin{equation}\label{eq:localized-envelope}
  A_\eta^{c_s,c_h}(x)
  \coloneqq
  c_s\big(\bar A(x)-R\big)_+^\theta
  +c_h\big(\bar A(x)-S_\eta\big)_+ .
\end{equation}
The first step to solve the design problem is to choose the deterministic growth indicator \(\bar A\).  The effective-linearity condition points to the growth score
\[
  G_\star(x)=\frac{\|b(x)\|}{1+\|x\|}.
\]
This score gives the scale on which the localized denominator should switch on and grow. Therefore, we take \(\bar A=G_\star.\) Since \(G_\star\) is expensive to evaluate during the production chain, we fit a proxy score for \(G_\star\) on a shifted log scale, choose \(R\) and \(S_\eta\) from pilot quantiles, and insert the proxy score into \eqref{eq:localized-envelope} to obtain the proxy-quantile envelope.

The proxy-quantile envelope fitted from a pilot run mainly describes the region visited by that run.  The practical denominator used in the reported experiments is therefore the local proxy-quantile denominator obtained from this envelope.  The stable trajectories and bounded diagnostics in the experiments show that this local denominator is already sufficient on the regions visited by the main chains.  This empirical stability, however, is not the same as a global Lyapunov guarantee: outside the calibrated region, the fitted proxy need not control all far-tail states.

For this reason we distinguish the implemented local denominator from the final denominator with a norm-polynomial tail-floor correction.  The tail-floor correction supplies the global effective-linearity input required by the companion deterministic-envelope Lyapunov theory without requiring full-gradient evaluations during production.  Thus the proxy-quantile construction explains the practical local design, while the tail-corrected version provides a formal global stability input when needed.

In implementation, the parameters \((q_R,q_S,\theta,\alpha,\eta)\) are treated as prescribed design parameters in the basic construction.  In the experiments, the fixed default choice already gives a competitive denominator.  Once the proxy score is fixed, a small calibration search over these parameters can further improve the denominator response in some cases.

\paragraph{Relation to existing tamed and stochastic gradient Langevin schemes.}
SGLD and stochastic gradient MCMC use stochastic-gradient oracles to construct scalable Langevin-type samplers \cite{WellingTeh2011,VollmerZygalakisTeh2016,ChenFoxGuestrin2014,MaChenFox2015}.  Existing analyses study discretization bias, stochastic-gradient noise, invariant-measure approximation, variance reduction, preconditioning, and nonasymptotic sampling error \cite{Raginsky2017,ZouXuGu2021,BrosseDurmusMoulines2018,Dubey2016,LiChenCarlsonCarin2016,Dalalyan2017,DurmusMoulines2017}.  For superlinear drifts, explicit Langevin or Euler-type schemes can lose stability, and tamed or adaptive modifications control large increments by changing the drift magnitude or the effective step response \cite{MattinglyStuartHigham2002,HutzenthalerJentzenKloeden2012,Brosse2019,LytrasMertikopoulos2024,LovasLytrasRasonyiSabanis2023}.  Related stabilization devices, such as gradient clipping, also illustrate that modifying stochastic-gradient increments can introduce a bias channel \cite{Koloskova2023}.  These works motivate denominator responses, but they do not address the proxy-calibration problem studied here: how to replace an exact growth score by a low-cost fitted score while preserving the level-set, threshold, and localized-envelope information needed by a deterministic denominator.

\paragraph{Contributions.}
This paper introduces a proxy-calibrated deterministic denominator construction.  Starting from the exact growth score
\[
  G_\star(x)=\frac{\|b(x)\|}{1+\|x\|},
\]
we fit a shifted-log proxy score from pilot samples and use pilot quantiles to calibrate the localized envelope.  This gives a low-cost production-stage denominator that retains the activation geometry of the exact-growth reference without evaluating \(G_\star\) at every step.

The main theoretical contribution is a comparison-and-transfer analysis for this construction.  We show how log-scale proxy accuracy yields level-set comparison, quantile stability, envelope-error control, modified-drift residual bounds, and stationary observable-error transfer.  The perturbation and concentration steps are based on standard Markov-chain and Poisson-equation tools \cite{RobertsRosenthalSchwartz1998,GlynnMeyn1996,Mitrophanov2005,RudolfSchweizer2018,GlynnOrmoneit2002,Paulin2015}, but the residuals are organized around proxy-induced denominator errors.  A norm-polynomial tail floor is added only to restore global growth control outside the pilot-calibrated region.

The experiments confirm that the fitted proxy tracks the \(G_\star\)-envelope closely, improves finite-chain behavior relative to cruder denominator responses, and keeps production-stage cost close to random-denominator baselines while avoiding full-gradient growth-score evaluations.

\section{Setup and notation}
\label{sec:setup}
We use the notation and deterministic-envelope convention of the companion localized-taming framework.

\subsection{Oracle convention and deterministic envelopes}
Let \(F_z:\R^d\to\R\) be an empirical potential or risk.  We use the drift convention
\begin{equation*}
  b(x)=b_z(x)=-\nabla F_z(x).
\end{equation*}
Throughout the paper, we work under the polynomial-growth dissipative Langevin setting used in the companion paper.  The drift is polynomially locally Lipschitz: for some constants \(L>0\) and \(r\ge0\),
\[
  \norm{b(x)-b(y)}
  \le
  L(1+\norm{x}^r+\norm{y}^r)\norm{x-y},
  \qquad x,y\in\R^d .
\]
It is also dissipative: for some constants \(m>0\) and \(b_0\ge0\),
\[
  \ip{x}{b(x)}
  \le
  b_0-m\norm{x}^{r+2},
  \qquad x\in\R^d .
\]

A stochastic-gradient oracle is written as
\begin{equation*}
  \widehat g_m(x,U)=\nabla F_z(x)+\zeta_m(x,U),\qquad
  \E[\zeta_m(x,U)\mid x]=0.
\end{equation*}
Equivalently, the stochastic drift oracle is \(-\widehat g_m(x,U)=b(x)-\zeta_m(x,U)\).

A deterministic envelope \(A_\eta\) is a state-dependent denominator term fixed before the current stochastic-gradient sample is drawn.  With such an envelope, the fixed-step chain used throughout the paper is
\begin{equation*}
  X_{k+1}=X_k-\eta\frac{\widehat g_m(X_k,U_k)}{1+\eta^\alpha A_\eta(X_k)}+\sqrt{2\eta/\beta}\,Z_{k+1}
  =X_k+\eta\frac{b(X_k)-\zeta_m(X_k,U_k)}{1+\eta^\alpha A_\eta(X_k)}+\sqrt{2\eta/\beta}\,Z_{k+1}.
\end{equation*}
In the schemes considered here, the drift is divided by \(1+\eta^\alpha A_\eta(x)\), so we define the retention factor
\begin{equation*}
  \Phi_{A_\eta}(x)=\frac{1}{1+\eta^\alpha A_\eta(x)}.
\end{equation*}
Thus \(\Phi_{A_\eta}(x)\) is the fraction of the original drift retained after taming.  Later we write \(\Phi_A\) when comparing generic fixed envelopes.

Because \(A_\eta(x)\) is fixed conditional on \(x\), the denominator preserves the conditional mean identity:
\begin{equation*}
  \E_U\left[\frac{\widehat g_m(x,U)}{1+\eta^\alpha A_\eta(x)}\mid x\right]
  =
  \frac{\nabla F_z(x)}{1+\eta^\alpha A_\eta(x)}.
\end{equation*}

We use the localized deterministic-envelope family from the companion paper~\cite{companion}.
\begin{definition}[Localized envelope family]
\label{def:localized-envelope-family}
Let \(\bar A:\R^d\to[0,\infty)\) be a deterministic growth indicator.  For parameters \(R\ge0\), \(0<\theta\le1\), an upper threshold \(S_\eta\ge R\), and activation indicators \(c_s,c_h\in\{0,1\}\), define
\begin{equation*}
  A_\eta^{c_s,c_h}(x)
  \coloneqq
  c_s\big(\bar A(x)-R\big)_+^\theta
  +c_h\big(\bar A(x)-S_\eta\big)_+ .
\end{equation*}
The local soft, local hard, and two-component localized envelopes are
\[
  A_\eta^{\rm soft} \coloneqq A_\eta^{1,0},\qquad
  A_\eta^{\rm hard} \coloneqq A_\eta^{0,1},\qquad
  A_\eta^{\rm two} \coloneqq A_\eta^{1,1}.
\]
\end{definition}

The companion stability result uses the following compatibility condition on the deterministic growth indicator:
\begin{equation}\label{eq:companion-growth-indicator-input}
  \frac{\norm{b(x)}}{1+\bar A(x)}
  \le C_{\bar A}(1+\norm{x}),
  \qquad
  \bar A(x)\le C_{\bar A}(1+\norm{x}^{r_{\bar A}}),
  \qquad x\in\R^d .
\end{equation}
With a hard-tail component \(c_h=1\) and a valid upper threshold, for example \(S_\eta\le C_S\eta^{-\alpha}\), this condition is the input used by Lemma~3.5 and Proposition~4.6 of the companion paper.

\subsection{The growth score and the deterministic growth indicator}
The localized family requires a deterministic growth indicator \(\bar A\).  As the ideal deterministic growth indicator, we choose
\begin{equation*}
  \bar A(x)=G_\star(x)
  \coloneqq \frac{\norm{b(x)}}{1+\norm{x}}
  =\frac{\norm{\nabla F_z(x)}}{1+\norm{x}}.
\end{equation*}
The reason for this choice is the effective-linearity goal.  A denominator should tame the drift only when \(\|b(x)\|\) is larger than the linear scale \(1+\|x\|\).  The growth score \(G_\star\) measures exactly this relative size.

This choice satisfies the growth-indicator compatibility in \eqref{eq:companion-growth-indicator-input}.  Indeed,
\begin{equation*}
  \frac{\norm{b(x)}}{1+G_\star(x)}
  =
  \frac{\norm{b(x)}(1+\norm{x})}
       {1+\norm{x}+\norm{b(x)}}
  \le 1+\norm{x}.
\end{equation*}
Under the standing polynomial-growth drift assumption, \(G_\star\) also has the corresponding polynomial upper growth.  Thus \(G_\star\) is a compatible deterministic growth indicator for the localized family when it is used as \(\bar A\).

Using the growth score \(G_\star\) as \(\bar A\) gives the \(G_\star\)-localized envelope.  Its level sets show where the denominator should switch on: larger values of \(G_\star\) mark states where the drift exceeds the linear-growth scale more strongly.  Its threshold excesses,
\begin{equation*}
  (G_\star(x)-R)_+,
  \qquad
  (G_\star(x)-S_\eta)_+,
\end{equation*}
determine the taming strength after activation in the soft and upper components of \eqref{eq:localized-envelope}.  The practical construction replaces this growth score by a proxy score \(\widehat G\).  The proxy is designed to match the level sets and threshold excesses of \(G_\star\).

We summarize the main notation used in the sequel in Table~\ref{tab:notation}.

\begin{table}[H]
\centering
\caption{Main notation.}
\label{tab:notation}
\small
\begin{tabular}{p{0.32\textwidth}p{0.60\textwidth}}
\toprule
Symbol & Meaning \\
\midrule
$F_z$ & empirical potential / risk \\
$b(x)=-\nabla F_z(x)$ & Langevin drift \\
$\widehat g_m(x,U)$ & stochastic-gradient oracle \\
$A_\eta^{c_s,c_h}$ & localized envelope family \\
$\Phi_A(x)$ & retention factor induced by an envelope \\
$G_\star(x)=\norm{b(x)}/(1+\norm{x})$ & growth score \\
$L_\star(x)=\log(G_\star(x)+\tau)$ & shifted log-scale target \\
$\widehat G(x)$ & fitted proxy score \\
$\phi(x)$ & proxy feature map \\
$\widehat R_{q_R},\widehat S_{q_S}$ & empirical pilot quantile thresholds \\
$\theta$ & soft-transition exponent \\
$A_{\rm loc}$ & proxy-quantile localized envelope \\
$D_{\eta,{\rm loc}}$ & local proxy denominator used in the main experiments \\
$D_{\eta,{\rm final}}$ & tail-corrected denominator used only when the norm-polynomial tail floor is added\\
\bottomrule
\end{tabular}
\end{table}

\section{Proxy-quantile deterministic denominator construction}
\label{sec:proxy-quantile-design}
\label{sec:proxy-geometry-quantile}

In Section~\ref{sec:setup}, we choose the deterministic growth indicator \(\bar A\) to be the growth score \(G_\star\).  This section replaces it by a fitted proxy \(\widehat G\), selects empirical quantile thresholds, and forms the local proxy-quantile envelope used by the implemented denominator.  A separate tail-floor correction is introduced later to supply a global stability input.  The proxy should preserve two quantities used by the localized denominator: level sets and threshold excesses.  We use a shifted positive-cone comparison to control both.

\subsection{Positive-cone order for nonnegative scores}

Both \(G_\star\) and \(\widehat G\) are nonnegative scores.  We compare them after adding a small shift \(\tau>0\).  For nonnegative scores \(f,g\), write
\begin{equation*}
  f\preceq_{\delta,\tau}g
  \quad\Longleftrightarrow\quad
  f(x)+\tau\le e^\delta(g(x)+\tau)
  \quad\text{on the region under consideration}.
\end{equation*}
If both inequalities hold, write \(f\asymp_{\delta,\tau}g\).  This is a multiplicative comparison after the positive shift.

\begin{lemma}[Cone comparability and log-ratio error]
\label{lem:cone-log-equivalence}
For nonnegative \(f,g\) and \(\tau>0\),
\begin{equation*}
  f\asymp_{\delta,\tau}g
\end{equation*}
if and only if
\begin{equation}
  \left|\log(f(x)+\tau)-\log(g(x)+\tau)\right|\le \delta
  \quad\text{on the region under consideration}.
  \label{eq:log-ratio-equivalence}
\end{equation}
Equivalently,
\begin{equation*}
  e^{-\delta}
  \le
  \frac{f(x)+\tau}{g(x)+\tau}
  \le
  e^\delta .
\end{equation*}
\end{lemma}
\begin{proof}
The two cone inequalities are
\begin{equation*}
  f(x)+\tau\le e^\delta(g(x)+\tau),
  \qquad
  g(x)+\tau\le e^\delta(f(x)+\tau).
\end{equation*}
Since both shifted scores are positive, division gives the two-sided ratio bound, and taking logarithms gives \eqref{eq:log-ratio-equivalence}.  The converse follows by exponentiation.
\end{proof}

\subsection{Sandwich bounds for level sets and excesses}
The \(\tau\)-shifted comparison below turns proxy-score accuracy into two estimates: one for upper-level sets and one for threshold excesses.  These are the two quantities needed to compare the localized denominators generated by \(G_\star\) and by \(\widehat G\).

\begin{proposition}[Level-set sandwich]
\label{prop:level-set-sandwich}
Assume \(\widehat G\asymp_{\delta,\tau}G_\star\).  Then for every \(s\ge0\),
\begin{equation*}
  \{G_\star>e^\delta(s+\tau)-\tau\}
  \subset
  \{\widehat G>s\}
  \subset
  \{G_\star>e^{-\delta}(s+\tau)-\tau\}.
\end{equation*}
\end{proposition}
\begin{proof}
If \(G_\star(x)>e^\delta(s+\tau)-\tau\), then
\begin{equation*}
  e^{-\delta}(G_\star(x)+\tau)-\tau>s.
\end{equation*}
The lower cone bound implies \(\widehat G(x)+\tau\ge e^{-\delta}(G_\star(x)+\tau)\), hence \(\widehat G(x)>s\).  Conversely, if \(\widehat G(x)>s\), then \(\widehat G(x)+\tau>s+\tau\).  The upper cone bound gives \(G_\star(x)>e^{-\delta}(s+\tau)-\tau\).
\end{proof}

The next statement records the corresponding comparison for the threshold excess.  This is the quantity that determines the denominator strength after activation.

\begin{proposition}[Positive-part sandwich]
\label{prop:positive-part-sandwich}
Assume \(\widehat G\asymp_{\delta,\tau}G_\star\) on a region.  Let \(\widehat s\ge0\) be a proxy threshold and let \(s^-\le s^+\) be two growth-score thresholds such that
\begin{equation}
  e^\delta(s^-+\tau)-\tau
  \le
  \widehat s
  \le
  e^{-\delta}(s^++\tau)-\tau .
  \label{eq:positive-part-threshold-compatibility}
\end{equation}
Then, on the same region,
\begin{equation}
  e^{-\delta}(G_\star(x)-s^+)_+
  \le
  (\widehat G(x)-\widehat s)_+
  \le
  e^\delta(G_\star(x)-s^-)_+ .
  \label{eq:positive-part-sandwich}
\end{equation}
\end{proposition}
\begin{proof}
The upper cone bound gives
\begin{equation*}
  \widehat G(x)\le e^\delta(G_\star(x)+\tau)-\tau .
\end{equation*}
Using the left inequality in \eqref{eq:positive-part-threshold-compatibility},
\begin{equation*}
  \widehat G(x)-\widehat s
  \le
  e^\delta(G_\star(x)+\tau)-\tau-\widehat s
  \le
  e^\delta(G_\star(x)-s^-).
\end{equation*}
Taking positive parts gives the upper bound in \eqref{eq:positive-part-sandwich}.  The lower cone bound gives
\begin{equation*}
  \widehat G(x)\ge e^{-\delta}(G_\star(x)+\tau)-\tau .
\end{equation*}
Using the right inequality in \eqref{eq:positive-part-threshold-compatibility},
\begin{equation*}
  \widehat G(x)-\widehat s
  \ge
  e^{-\delta}(G_\star(x)+\tau)-\tau-\widehat s
  \ge
  e^{-\delta}(G_\star(x)-s^+).
\end{equation*}
Taking positive parts gives the lower bound.
\end{proof}

\subsection{Log-scale proxy construction}
Lemma~\ref{lem:cone-log-equivalence} reduces cone comparison to uniform
approximation on the shifted log scale.  The target is
\[
  L_\star(x):=\log(G_\star(x)+\tau),
  \qquad \tau>0.
\]
We fit a \(\tau\)-shifted log-scale function \(h\) to approximate \(L_\star\), and then return to
the score scale by setting
\[
  G_h(x):=\bigl[\exp\{h(x)\}-\tau\bigr]_+.
\]
Thus  \(|h-L_\star|<\delta\) gives
\(G_h\asymp_{\delta,\tau}G_\star\); the positive part only enforces
nonnegativity.

For a chosen feature map \(\phi\), set
\begin{equation*}
  h_\omega(x)=\omega^\top\phi(x),
\end{equation*}
and
\begin{equation*}
  G_\omega(x):=\bigl[\exp\{\omega^\top\phi(x)\}-\tau\bigr]_+.
\end{equation*}
The feature map \(\phi\) should represent the growth scale of \(L_\star\).  It need not encode every detail of \(G_\star\).  The default low-cost choice is motivated by the following elementary bound.

By \eqref{eq:companion-growth-indicator-input} with \(\bar A=G_\star\), there
exist constants \(c_1,c_2>0\) and \(b\ge0\) such that
\[
  c_1
  \le
  G_\star(x)+\tau
  \le
  c_2(1+\norm{x})^b .
\]
Hence, with \(r(x)=\log(1+\norm{x})\),
\[
  \log c_1
  \le
  L_\star(x)
  \le
  \log c_2 + b r(x).
\]
Thus the low-cost features \(1,r,r^2\) provide a simple bracket for the
growth scale of \(L_\star\).

The design in this paper uses the following default feature map:
\begin{equation*}
  \phi(x)=\bigl(1,r(x),r(x)^2\bigr),
  \qquad r(x)=\log(1+\norm{x}).
\end{equation*}
The constant and \(\log(1+\norm{x})\) terms capture the log-growth bound implied by \eqref{eq:companion-growth-indicator-input}; the quadratic term gives extra flexibility for finite-sample fitting.
Other low-cost or problem-specific features can be used when the geometry of \(G_\star\) suggests a richer approximation class.

The coefficient vector is selected by least squares on pilot states \(X_1,\ldots,X_N\) using score measurements \(Y_i\) of the growth score:
\begin{equation}
  \widehat\omega
  \in
  \arg\min_\omega
  \sum_{i=1}^N
  \left(\omega^\top\phi(X_i)-\log(Y_i+\tau)\right)^2.
  \label{eq:log-proxy-fit}
\end{equation}
The proxy score used by the envelope is
\begin{equation*}
  \widehat G(x)=G_{\widehat\omega}(x)
  =\left[\exp\{\widehat\omega^\top\phi(x)\}-\tau\right]_+.
\end{equation*}

\subsection{A sufficient log-scale route to cone comparison}
\label{subsec:l2-logfit-cone}

The proxy score \(\widehat G\) is built from a pilot chain and a simple feature
map.  This construction is inexpensive, but it does not by itself imply global
cone comparison with \(G_\star\).  We therefore first state a local sufficient
condition on the pilot-explored region, prove that it implies cone comparison,
and then explain when the condition is satisfied for the default feature map
used in the experiments.

Let \(\Omega_0\) denote the pilot region, namely the part of state space covered
by the pilot run and used to fit the proxy score.  Let \(\nu\) be the pilot law
on \(\Omega_0\), and set
\begin{equation*}
  \mathcal H_\phi
  =
  \{h_\omega(x)=\omega^\top\phi(x):\omega\in\mathbb R^p\}.
\end{equation*}
The shifted log target is
\begin{equation*}
  L_\star(x)=\log(G_\star(x)+\tau),
  \qquad \tau>0.
\end{equation*}

\begin{assumption}[Local log-score calibration on the pilot region]
\label{ass:log-feature-calibration}
There exist constants \(C_\phi<\infty\), \(b_\phi\ge0\), and a function
\(h_0\in\mathcal H_\phi\) such that
\begin{align}
  \|u\|_{L^\infty(\Omega_0)}
  &\le C_\phi\|u\|_{L^2(\nu)},
  \qquad u\in\mathcal H_\phi,
  \label{eq:feature-norm-control}\\
  \|h_0-L_\star\|_{L^\infty(\Omega_0)}
  &\le b_\phi .
  \label{eq:feature-bias-control}
\end{align}
\end{assumption}

The first condition is a local coverage and nondegeneracy condition for the
chosen feature map.  It says that, within the finite-dimensional feature class,
small \(L^2(\nu)\) error cannot hide a large error somewhere in the pilot
region.  The second condition records the local approximation bias of the
chosen feature class for the shifted log target \(L_\star\).

\begin{proposition}[Local log-score control implies cone comparison]
\label{prop:l2-to-cone}
Assume Assumption~\ref{ass:log-feature-calibration}.  Let
\(\widehat h\in\mathcal H_\phi\) satisfy
\begin{equation*}
  \|\widehat h-h_0\|_{L^2(\nu)}\le \varepsilon .
\end{equation*}
Define
\begin{equation*}
  \widehat G(x)=\bigl[\exp\{\widehat h(x)\}-\tau\bigr]_+ .
\end{equation*}
Then, on \(\Omega_0\),
\begin{equation*}
  \|\widehat h-L_\star\|_{L^\infty(\Omega_0)}
  \le b_\phi+C_\phi\varepsilon .
\end{equation*}
Consequently, with \(\delta=b_\phi+C_\phi\varepsilon\),
\begin{equation*}
  e^{-\delta}(G_\star(x)+\tau)
  \le
  \widehat G(x)+\tau
  \le
  e^\delta(G_\star(x)+\tau),
  \qquad x\in\Omega_0.
\end{equation*}
That is, \(\widehat G\asymp_{\delta,\tau}G_\star\) on \(\Omega_0\).
\end{proposition}

\begin{proof}
By the triangle inequality and Assumption~\ref{ass:log-feature-calibration},
\begin{align*}
  \|\widehat h-L_\star\|_{L^\infty(\Omega_0)}
  &\le
  \|h_0-L_\star\|_{L^\infty(\Omega_0)}
  +
  \|\widehat h-h_0\|_{L^\infty(\Omega_0)}\\
  &\le
  b_\phi+C_\phi\varepsilon .
\end{align*}
Exponentiating gives
\[
  e^{-\delta}(G_\star+\tau)
  \le
  \exp\{\widehat h\}
  \le
  e^\delta(G_\star+\tau).
\]
Since
\[
  \widehat G+\tau=\max\{\exp\{\widehat h\},\tau\}
\]
and \(G_\star+\tau\ge \tau\), the same two-sided bound holds for
\(\widehat G+\tau\).
\end{proof}

\paragraph{Verification for the default feature.}
We now check the norm-control condition \eqref{eq:feature-norm-control} for the
default feature map used in the proxy construction.  Recall that
\begin{equation}
  \phi(x)=\bigl(1,r(x),r(x)^2\bigr),
  \qquad
  r(x)=\log(1+\norm{x}).
  \label{eq:default-log-radial-feature}
\end{equation}

\begin{lemma}[Norm control for the default log-radial feature]
\label{lem:default-feature-norm-control}
Assume that \(\Omega_0\) is bounded, so that
\[
  0\le r(x)\le R_0,
  \qquad x\in\Omega_0.
\]
Let
\[
  \psi(r)=(1,r,r^2)^\top,
  \qquad
  M_\nu=\int_{\Omega_0}\psi(r(x))\psi(r(x))^\top\,\nu(dx).
\]
If \(M_\nu\) is positive definite, then there exists \(C_\phi<\infty\) such that
\[
  \|u\|_{L^\infty(\Omega_0)}
  \le
  C_\phi\|u\|_{L^2(\nu)},
  \qquad
  u\in\operatorname{span}\{1,r,r^2\}.
\]
In particular, one may take
\[
  C_\phi
  =
  \frac{\sup_{0\le r\le R_0}\|\psi(r)\|}
       {\lambda_{\min}(M_\nu)^{1/2}} .
\]
\end{lemma}

\begin{proof}
Write \(u(x)=a^\top\psi(r(x))\).  Then
\[
  \|u\|_{L^\infty(\Omega_0)}
  \le
  \sup_{0\le r\le R_0}\|\psi(r)\|\,\|a\|.
\]
On the other hand,
\[
  \|u\|_{L^2(\nu)}^2
  =
  a^\top M_\nu a
  \ge
  \lambda_{\min}(M_\nu)\|a\|^2 .
\]
Combining the two estimates gives the claim.
\end{proof}

The condition \(M_\nu\succ0\) is a nondegenerate radial-coverage condition for
the pilot chain.  For the default feature, no directional coverage is required
by this particular norm-control step, because the feature is radial.  If the
pilot samples are concentrated in an extremely narrow radial band, then
\(M_\nu\) becomes nearly singular and the constant \(C_\phi\) becomes large.

\paragraph{Interpreting the approximation bias for polynomial-tail scores.}
We next record a sufficient condition for the approximation-bias condition
\eqref{eq:feature-bias-control}.  The point is not that polynomial upper
growth alone forces the default log-radial proxy to be accurate.  Rather, when
the shifted growth score has a nondegenerate polynomial tail scale on the pilot
region, the shifted log target is well represented by log-radial features.

\begin{lemma}[Log-radial bias for polynomial-growth scores]
\label{lem:poly-growth-log-bias}
Assume that, on \(\Omega_0\), there exist \(p\ge0\) and constants
\(0<c_1\le c_2<\infty\) such that
\begin{equation}
  c_1(1+\norm{x})^p
  \le
  G_\star(x)+\tau
  \le
  c_2(1+\norm{x})^p .
  \label{eq:radially-comparable-polynomial-score}
\end{equation}
Let \(r(x)=\log(1+\norm{x})\) and
\[
  \mathcal H_\phi=\operatorname{span}\{1,r,r^2\}.
\]
Then there exists \(h_0\in\mathcal H_\phi\) such that
\[
  \|h_0-L_\star\|_{L^\infty(\Omega_0)}
  \le
  \log(c_2/c_1).
\]
In particular, one may take
\[
  h_0(x)=\log c_1+p\,r(x).
\]
\end{lemma}

\begin{proof}
By \eqref{eq:radially-comparable-polynomial-score},
\[
  \log c_1+p\,r(x)
  \le
  L_\star(x)
  \le
  \log c_2+p\,r(x),
  \qquad x\in\Omega_0 .
\]
Taking \(h_0(x)=\log c_1+p\,r(x)\) gives
\[
  0
  \le
  L_\star(x)-h_0(x)
  \le
  \log(c_2/c_1),
  \qquad x\in\Omega_0,
\]
which proves the claim.
\end{proof}

For instance, if on \(\Omega_0\) there exist constants
\(0<c_b\le C_b<\infty\) and \(q\ge2\) such that
\[
  c_b(1+\norm{x})^{q-1}
  \le
  \norm{b(x)}
  \le
  C_b(1+\norm{x})^{q-1},
\]
then
\[
  c_b(1+\norm{x})^{q-2}
  \le
  G_\star(x)
  \le
  C_b(1+\norm{x})^{q-2}.
\]
Thus the default feature class contains the required affine log-radial scale,
while the additional \(r(x)^2\) term gives finite-range fitting flexibility.
This condition is satisfied by standard coercive polynomial-tail risks, not
only by radial ones.  It is enough that, in the tail,
\[
  b(x)=b_q(x)+o(\norm{x}^{q-1}),
  \qquad
  \norm{b_q(x)}\ge c_q\norm{x}^{q-1},
\]
where \(b_q\) is the leading polynomial drift.  Then
\(\norm{b(x)}\asymp \norm{x}^{q-1}\), and the shifted growth score inherits
the corresponding radial scale whenever it is constructed from this drift
growth.  This covers coercive radial polynomial potentials, separable
coercive polynomial potentials, and regularized empirical risks whose
coercive polynomial regularizer dominates lower-order empirical-gradient
terms.  It excludes degenerate polynomial directions where the leading drift
vanishes, which is consistent with its role as a tail-calibration condition.

The default log-radial feature is only one low-cost choice.  If the shifted
growth score has substantial nonradial structure on the pilot region, this
appears as a larger calibration bias \(b_\phi\).  The same argument applies
after adding simple problem-specific features, such as block norms,
directional quadratic ratios, layer-wise norms, or other low-dimensional
features suggested by the risk structure.  Only the approximation bias
\(b_\phi\) and the finite-dimensional norm-control constant \(C_\phi\)
change.

If Assumption~\ref{ass:log-feature-calibration} is unavailable, an \(L^2\)
log-scale error still gives high-mass control under the pilot law:
\begin{equation*}
  \nu\{|\widehat h-L_\star|>\delta\}
  \le
  \frac{\|\widehat h-L_\star\|_{L^2(\nu)}^2}{\delta^2}.
\end{equation*}
This weaker statement is useful as a diagnostic, but the clean level-set and
positive-part comparisons used below rely on the cone comparison.

\subsection{Quantile thresholds and threshold errors}
\label{subsec:quantile-threshold-errors}

Let \(\mathsf Q_q(Z)\) denote the lower \(q\)-quantile of a real-valued random
variable \(Z\):
\[
  \mathsf Q_q(Z):=\inf\{t\in\mathbb R:\mathbb P(Z\le t)\ge q\}.
\]
All population quantiles in this subsection are taken under the pilot law
\(\nu\).  We use the local cone comparison from
Proposition~\ref{prop:l2-to-cone} on \(\Omega_0\); hence the quantile
statements concern the pilot law on the pilot region.

If \(G_\star\) were available on the pilot law, the reference population
thresholds would be
\begin{equation*}
  R_{q_R}^\star=\mathsf Q_{q_R}(G_\star),
  \qquad
  S_{q_S}^\star=\mathsf Q_{q_S}(G_\star),
  \qquad 0<q_R<q_S<1.
\end{equation*}
Here \(R_{q_R}^\star\) marks the boundary between the safe and intermediate
regions, while \(S_{q_S}^\star\) marks the onset of the tail region.

The practical construction replaces these reference thresholds in two steps:
\[
  \mathsf Q_q(G_\star)
  \quad\longrightarrow\quad
  \mathsf Q_q(\widehat G)
  \quad\longrightarrow\quad
  \widehat{\mathsf Q}_{N,q}(\widehat G).
\]
The first arrow is the population proxy-threshold error caused by replacing
\(G_\star\) with \(\widehat G\).  The second arrow is the empirical
pilot-threshold error caused by estimating the proxy quantile from the Markov
pilot chain.  The finite-pilot bound below is stated for a fixed fitted score
\(\widehat G\), which is exact under sample splitting.  If the same pilot run
is reused for fitting and threshold estimation, the deterministic quantile
argument is unchanged, but the concentration input must be justified for the
resulting data-dependent fitted score.

\paragraph{Population proxy-threshold error.}
This error is the difference between the population thresholds of the growth
score and those of the proxy score.  Cone comparison controls it before any
pilot sampling error is added.  We first record two elementary properties of
lower quantiles.
\begin{lemma}[Monotonicity and affine equivariance of lower quantiles]
\label{lem:quantile-basic}
Let \(X,Y\) be real random variables under the pilot law.  If \(X\le Y\) almost
surely, then
\begin{equation*}
  \mathsf Q_q(X)\le \mathsf Q_q(Y)
\end{equation*}
for every \(q\in(0,1)\).  If \(a>0\) and \(c\in\R\), then
\begin{equation*}
  \mathsf Q_q(aX+c)=a\mathsf Q_q(X)+c.
\end{equation*}
\end{lemma}
\begin{proof}
If \(X\le Y\), then \(\{Y\le t\}\subset\{X\le t\}\), so
\(F_Y(t)\le F_X(t)\).  Hence the lower \(q\)-quantile of \(Y\) cannot lie to
the left of that of \(X\).  The affine identity follows from
\(F_{aX+c}(t)=F_X((t-c)/a)\) for \(a>0\).
\end{proof}

\begin{proposition}[Cone comparison controls population proxy quantiles]
\label{prop:quantile-sandwich}
Assume \(\widehat G\asymp_{\delta,\tau}G_\star\).  Then for every
\(q\in(0,1)\),
\begin{equation}
  e^{-\delta}\bigl(\mathsf Q_q(G_\star)+\tau\bigr)-\tau
  \le
  \mathsf Q_q(\widehat G)
  \le
  e^\delta\bigl(\mathsf Q_q(G_\star)+\tau\bigr)-\tau.
  \label{eq:quantile-sandwich}
\end{equation}
In particular, the population proxy thresholds at levels \(q_R\) and \(q_S\)
are controlled by the corresponding growth-score thresholds up to the shifted
multiplicative factors \(e^{\pm\delta}\).
\end{proposition}
\begin{proof}
The cone inequalities imply
\begin{equation*}
  e^{-\delta}(G_\star+\tau)-\tau
  \le
  \widehat G
  \le
  e^\delta(G_\star+\tau)-\tau.
\end{equation*}
The claim follows from Lemma~\ref{lem:quantile-basic} and affine equivariance.
\end{proof}

\paragraph{Empirical pilot-threshold error.}
This error is the difference between the population proxy threshold and the
empirical threshold computed from the pilot chain.  Let
\(\widehat{\mathsf Q}_{N,q}\) denote the empirical lower \(q\)-quantile of
pilot scores.  For real numbers \(z_1,\ldots,z_N\), define
\[
  \widehat{\mathsf Q}_{N,q}(z_1,\ldots,z_N)
  :=
  z_{(\lceil Nq\rceil)},
  \qquad q\in(0,1),
\]
where \(z_{(1)}\le\cdots\le z_{(N)}\) are the order statistics.  For fixed
pilot states \(X_1,\ldots,X_N\), write
\[
  \widehat{\mathsf Q}_{N,q}(f)
  :=
  \widehat{\mathsf Q}_{N,q}\bigl(f(X_1),\ldots,f(X_N)\bigr).
\]
Along the pilot trajectory, we evaluate
\(\widehat G(X_1),\ldots,\widehat G(X_N)\) and set
\begin{equation*}
  \widehat R_{q_R}=\widehat{\mathsf Q}_{N,q_R}(\widehat G),
  \qquad
  \widehat S_{q_S}=\widehat{\mathsf Q}_{N,q_S}(\widehat G),
  \qquad 0<q_R<q_S<1.
\end{equation*}

Finite pilot samples replace the population quantile of \(\widehat G(X)\) by
its empirical version along the pilot chain.  In this paragraph the fitted
score \(\widehat G\) is fixed.  Let
\[
  Z=\widehat G(X),\qquad X\sim\nu,
\]
let \(F_Z\) be the distribution function of \(Z\), and write
\[
  z_q=\mathsf Q_q(Z).
\]
For pilot states \(X_1,\ldots,X_N\), define
\[
  Z_i=\widehat G(X_i),
  \qquad
  \widehat F_N(t)=\frac1N\sum_{i=1}^N{\bf 1}\{Z_i\le t\},
  \qquad
  \widehat z_{N,q}=\widehat{\mathsf Q}_{N,q}(Z_1,\ldots,Z_N).
\]
The next lemma proves the empirical-quantile bound
\[
  |\widehat z_{N,q}-z_q|
  \le
  \frac{2}{c_q}\sqrt{\frac{\log(4/\zeta)}{N_{\rm eff}}}.
\]
The proof only needs two endpoint checks: \(\widehat F_N\) must remain below
\(q\) at \(z_q-r_N\), and it must have reached \(q\) by \(z_q+r_N\).

\begin{assumption}[Nondegenerate crossing at the pilot quantile]
\label{ass:quantile-margin}
For the level \(q\) under consideration, there are constants \(c_q>0\) and
\(r_q>0\) such that for every \(0\le r\le r_q\),
\begin{equation}
  F_Z(z_q-r)\le q-c_q r,
  \qquad
  F_Z(z_q+r)\ge q+c_q r.
  \label{eq:quantile-margin}
\end{equation}
\end{assumption}
The assumption says that the pilot-score distribution is not flat near the
selected threshold.  Moving a distance \(r\) to the right of \(z_q\) increases
the CDF by at least \(c_q r\), while moving the same distance to the left
decreases the CDF by at least \(c_q r\).

A standard sufficient condition is a positive local density at the threshold.
Suppose that \(Z\) has a density \(f_Z\), and that for some \(m_q>0\),
\begin{equation*}
  f_Z(t)\ge m_q
  \qquad
  \text{for all }t\in(z_q-\rho_q,z_q+\rho_q)
\end{equation*}
for some \(\rho_q>0\).  If \(F_Z(z_q)=q\), then for every
\(0\le r\le \rho_q\),
\begin{align*}
  F_Z(z_q+r)-q
  &=
  F_Z(z_q+r)-F_Z(z_q)
  =
  \int_{z_q}^{z_q+r} f_Z(t)\,dt
  \ge m_q r,\\
  q-F_Z(z_q-r)
  &=
  F_Z(z_q)-F_Z(z_q-r)
  =
  \int_{z_q-r}^{z_q} f_Z(t)\,dt
  \ge m_q r.
\end{align*}
Thus \eqref{eq:quantile-margin} holds after taking \(r_q\le\rho_q\) and any
\(c_q<m_q\).

\begin{assumption}[Markov pilot concentration at quantile endpoints]
\label{ass:markov-pilot-endpoint-concentration}
Let \(X_1,\ldots,X_N\) be the pilot Markov chain and let
\[
  Z_i=\widehat G(X_i),
  \qquad i=1,\ldots,N,
\]
where the fitted score \(\widehat G\) is treated as fixed.  There is an
effective sample size \(N_{\rm eff}\) such that, for every fixed threshold
\(t\) and every \(u>0\),
\begin{equation}
  \mathbb P\left(
    \left|
      \widehat F_N(t)-F_Z(t)
    \right|>u
  \right)
  \le
  2\exp\{-N_{\rm eff}u^2\}.
  \label{eq:markov-endpoint-concentration}
\end{equation}
\end{assumption}

Assumption~\ref{ass:markov-pilot-endpoint-concentration} is the Markov-chain
concentration input needed to control the pilot empirical quantiles.  Under
uniform ergodicity or suitable mixing conditions, such bounds follow from
Markov-chain concentration inequalities; see
\cite{GlynnOrmoneit2002,Paulin2015}.  The effective sample size
\(N_{\rm eff}\) summarizes the loss caused by dependence in the pilot chain. In
the proof below, the bound is applied only at the two endpoints \(z_q-r_N\) and
\(z_q+r_N\).

\begin{lemma}[Pilot quantile stability]
\label{lem:pilot-quantile-stability}
Assume Assumptions~\ref{ass:quantile-margin} and~\ref{ass:markov-pilot-endpoint-concentration}.  Let
\[
  \widehat z_{N,q}
  :=
  \inf\{z\in\mathbb R:\widehat F_N(z)\ge q\}
\]
be the lower empirical \(q\)-quantile.  Define
\[
  u_N:=\sqrt{\frac{\log(4/\zeta)}{N_{\rm eff}}},
  \qquad
  r_N:=\frac{2u_N}{c_q}.
\]
If \(r_N\le r_q\), then with probability at least \(1-\zeta\),
\begin{equation}
  \widehat z_{N,q}\in[z_q-r_N,z_q+r_N].
  \label{eq:empirical-quantile-error}
\end{equation}
Equivalently,
\[
  |\widehat z_{N,q}-z_q|
  \le
  \frac{2}{c_q}
  \sqrt{\frac{\log(4/\zeta)}{N_{\rm eff}}}.
\]
\end{lemma}
\begin{proof}
By Assumption~\ref{ass:quantile-margin} and the definition of \(r_N\),
\[
  F_Z(z_q-r_N)
  \le
  q-c_qr_N
  =
  q-2u_N,
\]
and
\[
  F_Z(z_q+r_N)
  \ge
  q+c_qr_N
  =
  q+2u_N .
\]
Apply Assumption~\ref{ass:markov-pilot-endpoint-concentration} to the two fixed
thresholds \(z_q-r_N\) and \(z_q+r_N\).  By a union bound, with probability at
least \(1-\zeta\), both inequalities
\[
  |\widehat F_N(z_q-r_N)-F_Z(z_q-r_N)|\le u_N,
  \qquad
  |\widehat F_N(z_q+r_N)-F_Z(z_q+r_N)|\le u_N
\]
hold.  On this event,
\[
  \widehat F_N(z_q-r_N)
  \le
  F_Z(z_q-r_N)+u_N
  \le
  q-u_N
  <q,
\]
and
\[
  \widehat F_N(z_q+r_N)
  \ge
  F_Z(z_q+r_N)-u_N
  \ge
  q+u_N
  >q.
\]

By the definition of the lower empirical quantile,
\[
  \widehat z_{N,q}
  =
  \inf\{z\in\mathbb R:\widehat F_N(z)\ge q\}.
\]
Since \(\widehat F_N(z_q-r_N)<q\), the empirical CDF has not reached level
\(q\) at \(z_q-r_N\).  Hence the first point at which it reaches level \(q\)
cannot lie to the left of \(z_q-r_N\), and therefore
\[
  \widehat z_{N,q}\ge z_q-r_N.
\]
Since \(\widehat F_N(z_q+r_N)>q\), the empirical CDF has reached level \(q\) by
\(z_q+r_N\).  Hence the first point at which it reaches level \(q\) cannot lie
to the right of \(z_q+r_N\), and therefore
\[
  \widehat z_{N,q}\le z_q+r_N.
\]
Combining the two bounds gives
\[
  \widehat z_{N,q}\in[z_q-r_N,z_q+r_N],
\]
as claimed.
\end{proof}

\begin{remark}[Data reuse in the pilot run]
\label{rem:markov-pilot-quantile-input}
Assumption~\ref{ass:markov-pilot-endpoint-concentration} is stated for a fixed
fitted score \(\widehat G\).  This is immediate under sample splitting.  If
the same pilot run is used for fitting and threshold estimation, then the
indicators also depend on the data through \(\widehat G\).  This does not
change Lemma~\ref{lem:pilot-quantile-stability}, but it requires the
corresponding endpoint concentration input for the data-dependent fitted score.
\end{remark}

\subsection{Proxy-quantile envelope}
\label{subsec:proxy-quantile-envelope}

The threshold analysis above supplies the empirical proxy thresholds used in the
localized envelope.  If \(G_\star\) were used directly, the corresponding
reference envelope would be
\begin{equation*}
  A_{q_R,q_S}^\star(x)
  =(G_\star(x)-R_{q_R}^\star)_+^\theta+(G_\star(x)-S_{q_S}^\star)_+ .
\end{equation*}
The proxy-quantile envelope applies the same localized-envelope form to
\(\widehat G\) and its empirical thresholds.

The fixed proxy-quantile convention used below is
\begin{equation*}
  q_R=0.70,
  \qquad
  q_S=0.99,
  \qquad
  \theta=\frac12.
\end{equation*}
This choice places most pilot states in the safe region and uses the upper
pilot quantile to mark the onset of the tail region.  The exponent
\(\theta=1/2\) is the soft parameter in the first positive-part term: it makes
the localized-envelope response sublinear after the safe--intermediate boundary
rather than hard immediately at that boundary.  Substituting
\(\bar A=\widehat G\), \(\widehat R_{q_R}\), \(\widehat S_{q_S}\), and
\(c_s=c_h=1\) into the localized family \eqref{eq:localized-envelope} gives the proxy-quantile localized envelope
\begin{equation*}
 A_{\rm loc}(x):= A_{\widehat G,\widehat R_{q_R},\widehat S_{q_S},\theta}(x)
  =
  (\widehat G(x)-\widehat R_{q_R})_+^\theta
  +
  (\widehat G(x)-\widehat S_{q_S})_+,
  \qquad 0<\theta\le1.
\end{equation*}
The threshold \(\widehat R_{q_R}\) marks the boundary between the safe and
intermediate regions, while \(\widehat S_{q_S}\) marks the onset of the tail
region.  The first term gives the soft activation of the localized envelope
after \(\widehat R_{q_R}\), and the second term adds the tail-region
contribution after \(\widehat S_{q_S}\).  The exponent \(\theta\) is only the
soft parameter in the first term; it does not change the two pilot thresholds.

Algorithm~\ref{alg:proxy-quantile-calibration} summarizes the proxy-quantile construction.  It fixes the proxy score and its quantile thresholds before any later calibration step.

\begin{algorithm}[H]
\caption{Proxy--quantile envelope construction}
\label{alg:proxy-quantile-calibration}
\KwInput{Pilot states \(X_1,\ldots,X_N\); quantile levels \(q_R<q_S\); floor \(\tau>0\); soft exponent \(\theta\) (default \(1/2\)); feature map \(\phi:\R^d\to\R^p\).}
\KwOutput{Proxy score \(\widehat G\), thresholds \(\widehat R,\widehat S\), and local deterministic envelope \(A_{\rm loc}\).}
\For{\(i=1,\ldots,N\)}{
  Compute or estimate the growth-score label
  \(Y_i\approx G_\star(X_i)=\|b(X_i)\|/(1+\|X_i\|)\)\;}
Construct a log-scale switching score by least squares:
\(\widehat\omega\in\arg\min_\omega\sum_{i=1}^N(\omega^\top\phi(X_i)-\log(Y_i+\tau))^2\)\;
Define \(\widehat G(x)=[\exp(\widehat\omega^\top\phi(x))-\tau]_+\)\;
Compute pilot proxy scores \(\widehat G_i=\widehat G(X_i)\)\;
Set \(\widehat R\leftarrow\widehat{\mathsf Q}_{N,q_R}(\widehat G_1,\ldots,\widehat G_N)\) and
\(\widehat S\leftarrow\widehat{\mathsf Q}_{N,q_S}(\widehat G_1,\ldots,\widehat G_N)\)\;
Set \(A_{\rm loc}(x)\leftarrow(\widehat G(x)-\widehat R)_+^\theta+(\widehat G(x)-\widehat S)_+\)\;
\KwReturn{\(\widehat G,\widehat R,\widehat S,A_{\rm loc}\)}\;
\end{algorithm}

\subsection{Final denominator with tail-floor correction}
\label{sec:tail-floor-correction}

The proxy-quantile construction gives the localized envelope
\begin{equation*}
  A_{\rm loc}(x)
  =
  (\widehat G(x)-\widehat R_{q_R})_+^\theta
  +
  (\widehat G(x)-\widehat S_{q_S})_+,
  \qquad 0<\theta\le1,
\end{equation*}
and the local denominator
\begin{equation}
  D_{\eta,{\rm loc}}(x)=1+\eta^\alpha A_{\rm loc}(x).
  \label{eq:local-proxy-denominator}
\end{equation}
This is the denominator reported under the label \emph{proxy-quantile} in the
main experiments.  Thus the main numerical results test the local proxy design
itself.  The stable trajectories observed there show that
\(D_{\eta,{\rm loc}}\) is sufficient on the regions visited by the chains.  They
do not, by themselves, prove a global Lyapunov condition on all of \(\R^d\).

For a global stability statement, the deterministic-envelope Lyapunov result
of \cite{companion} requires the scaled drift to have at most linear growth.  In
the present notation, the required input is
\begin{equation}
  \frac{\norm{b(x)}}{D_{\eta,{\rm final}}(x)}
  \le
  C_\eta(1+\norm{x}),
  \qquad x\in\R^d .
  \label{eq:final-effective-linearity}
\end{equation}
The tail floor enforces this input through a fixed norm-polynomial tail
indicator.  It does not use a fresh full-gradient evaluation during production,
and it does not require the fitted proxy score to extrapolate correctly in the
unobserved far tail.

Fix a polynomial exponent \(\kappa_{\rm tail}>0\) and define
\begin{equation}
  P_{\rm tail}(x)=1+\norm{x}^{\kappa_{\rm tail}} .
  \label{eq:polynomial-tail-indicator}
\end{equation}
In the tail diagnostics we use \(\kappa_{\rm tail}=2\), so
\(P_{\rm tail}(x)=1+\norm{x}^2\).  More generally, under a polynomial drift
bound \(\norm{b(x)}\le C_b(1+\norm{x}^m)\), a choice
\(\kappa_{\rm tail}\ge m-1\) gives a polynomial upper control of the growth
score \(G_\star(x)=\norm{b(x)}/(1+\norm{x})\).

Fix two pilot quantile levels
\[
  q_{\rm lin}<q_{\rm tail}<1 .
\]
Let
\begin{equation}
  C_{\rm lin}
  :=
  \rho_{\rm lin}\,
  \widehat{\mathsf Q}_{N,q_{\rm lin}}(P_{\rm tail}),
  \qquad
  T_{\rm tail}
  :=
  \widehat{\mathsf Q}_{N,q_{\rm tail}}(P_{\rm tail}),
  \label{eq:tail-threshold-and-linear-constant}
\end{equation}
where \(\rho_{\rm lin}>0\) is a fixed multiplier and the empirical quantiles are
computed from the pilot values \(P_{\rm tail}(X_1),\ldots,P_{\rm tail}(X_N)\).
In the tail diagnostics we use \(q_{\rm lin}=0.95\),
\(q_{\rm tail}=0.995\), and \(\rho_{\rm lin}=1\).  Here \(T_{\rm tail}\) is the
far-tail activation threshold, and \(C_{\rm lin}\) is the polynomial-tail level
used after activation.

The final denominator is obtained by flooring the local denominator in the far
tail:
\begin{equation}
  D_{\eta,{\rm final}}(x)
  =
  \begin{cases}
    D_{\eta,{\rm loc}}(x),
      & P_{\rm tail}(x)\le T_{\rm tail},\\[2mm]
    \max\left\{D_{\eta,{\rm loc}}(x),\dfrac{P_{\rm tail}(x)}{C_{\rm lin}}\right\},
      & P_{\rm tail}(x)>T_{\rm tail}.
  \end{cases}
  \label{eq:tail-final-denominator}
\end{equation}
Equivalently, define the tail floor
\begin{equation}
  D_{\eta,{\rm tail}}(x)
  :=
  \begin{cases}
    1,
      & P_{\rm tail}(x)\le T_{\rm tail},\\[2mm]
    \dfrac{P_{\rm tail}(x)}{C_{\rm lin}},
      & P_{\rm tail}(x)>T_{\rm tail},
  \end{cases}
  \label{eq:tail-denominator}
\end{equation}
so that
\[
  D_{\eta,{\rm final}}
  =
  \max\{D_{\eta,{\rm loc}},D_{\eta,{\rm tail}}\}.
\]
Below \(T_{\rm tail}\), the tail floor is inactive.  Above \(T_{\rm tail}\), it
raises the denominator when the local denominator is smaller than
\(P_{\rm tail}/C_{\rm lin}\).

To use this polynomial floor for a global stability statement, we require that
\(P_{\rm tail}\) upper controls the growth score at the scale relevant for the
Lyapunov argument: there is a constant \(C_{\rm tail}\) such that
\begin{equation}
  G_\star(x)
  \le
  C_{\rm tail} P_{\rm tail}(x),
  \qquad x\in\R^d .
  \label{eq:polynomial-tail-domination}
\end{equation}
This condition is not needed for running the local proxy-quantile denominator.
It is the extra input used to turn the polynomial tail floor into a global
Lyapunov bound.  Unlike proxy-score extrapolation in the far tail, it follows directly from a
polynomial drift-growth bound once \(\kappa_{\rm tail}\) is chosen large enough.

Under \eqref{eq:polynomial-tail-domination}, the final denominator satisfies
\eqref{eq:final-effective-linearity}.  If \(P_{\rm tail}(x)\le T_{\rm tail}\),
then \(D_{\eta,{\rm final}}(x)\ge1\), and hence
\[
  \frac{\norm{b(x)}}{D_{\eta,{\rm final}}(x)}
  \le
  \norm{b(x)}
  =
  G_\star(x)(1+\norm{x})
  \le
  C_{\rm tail}T_{\rm tail}(1+\norm{x}).
\]
If \(P_{\rm tail}(x)>T_{\rm tail}\), then
\(D_{\eta,{\rm final}}(x)\ge P_{\rm tail}(x)/C_{\rm lin}\), so
\[
  \frac{\norm{b(x)}}{D_{\eta,{\rm final}}(x)}
  \le
  C_{\rm tail}C_{\rm lin}(1+\norm{x}).
\]
Therefore \eqref{eq:final-effective-linearity} holds with
\[
  C_\eta=C_{\rm tail}\max\{T_{\rm tail},C_{\rm lin}\}.
\]
In this sense, \(D_{\eta,{\rm final}}\) is the local proxy denominator equipped
with a cheap norm-polynomial far-tail floor that enforces the scaled-drift
growth input under \eqref{eq:polynomial-tail-domination}.

A useful diagnostic is the tail-activation fraction
\begin{equation}
  \frac1T\sum_{k=1}^T
  {\bf 1}\{D_{\eta,{\rm tail}}(X_k)>D_{\eta,{\rm loc}}(X_k)\}.
  \label{eq:tail-activation-diagnostic}
\end{equation}
When this fraction is small, the averages are driven by the local proxy-quantile
denominator, and the tail floor acts as a far-tail correction.  When it is large,
the tail-floor correction is part of the tested denominator and should be
reported as such.  Section~\ref{sec:exp4-tail-floor} checks this point by
measuring both the activation frequency and the conditional reduction in the
effective growth score on states where the tail floor is active.

\begin{proposition}[Stability input for the tail-corrected denominator]
\label{prop:tail-corrected-stability}
Assume the standing drift assumptions and the polynomial-tail domination
condition \eqref{eq:polynomial-tail-domination}.  After the pilot fit and
threshold calibration, define \(D_{\eta,{\rm loc}}\), \(D_{\eta,{\rm tail}}\),
and \(D_{\eta,{\rm final}}\) by \eqref{eq:local-proxy-denominator},
\eqref{eq:tail-denominator}, and \eqref{eq:tail-final-denominator}.  Then
\(D_{\eta,{\rm final}}\) is state-deterministic and satisfies the global
effective-linearity input \eqref{eq:final-effective-linearity}.  Assume moreover
that \(D_{\eta,{\rm loc}}\) has the polynomial upper growth needed in the
deterministic-envelope Lyapunov estimate of
\cite[Proposition~4.6]{companion}.  Then the exact-gradient kernel with
denominator \(D_{\eta,{\rm final}}\) satisfies the Lyapunov drift estimate of
\cite[Proposition~4.6]{companion}, admits invariant measures, and has the
uniform moment bounds of \cite[Theorem~4.4]{companion}.

If the stochastic-gradient version is used, the same conclusion holds under the
oracle moment and envelope-compatible increment assumptions in
\cite[Assumption~7.1]{companion}, by the stochastic-gradient stability result
\cite[Proposition~7.2]{companion}.
\end{proposition}

\begin{proof}
After calibration, the fitted score \(\widehat G\), the local thresholds
\(\widehat R_{q_R}\), \(\widehat S_{q_S}\), the polynomial-tail threshold
\(T_{\rm tail}\), and the level \(C_{\rm lin}\) are fixed.  Hence
\(D_{\eta,{\rm loc}}\) and \(D_{\eta,{\rm tail}}\) are deterministic functions of
the current state.  Their pointwise maximum \(D_{\eta,{\rm final}}\) is also
state-deterministic.

The bound \eqref{eq:final-effective-linearity} gives the effective-linearity
input required in \cite[Proposition~4.6]{companion}.  It was proved above from
the polynomial-tail domination condition.  The other denominator input in that
proposition is polynomial upper growth.  This holds for the local part by
assumption.  The tail floor is controlled by \(P_{\rm tail}\), a fixed
polynomial in \(\norm{x}\).  Taking a pointwise maximum preserves positivity and
the required polynomial upper growth.

Therefore the exact-gradient deterministic-envelope Lyapunov estimate
\cite[Proposition~4.6]{companion} applies to the kernel with denominator
\(D_{\eta,{\rm final}}\).  The invariant-measure existence and the uniform
moment bounds then follow from \cite[Theorem~4.4]{companion}.

For the stochastic-gradient kernel, the denominator remains fixed conditional on
the state after calibration.  Under the oracle moment and compatibility
conditions in \cite[Assumption~7.1]{companion}, the stochastic-gradient
extension \cite[Proposition~7.2]{companion} gives the same Lyapunov drift
estimate.  The invariant-measure and moment conclusions follow in the same way.
\end{proof}

\subsection{Design summary}
\label{sec:design-summary}

The construction has two parts.  The main proxy-quantile design is
\begin{equation*}
  G_\star
  \quad\Longrightarrow\quad
  \log(G_\star+\tau)
  \quad\Longrightarrow\quad
  \widehat G
  \quad\Longrightarrow\quad
  A_{\rm loc}
  =
  A_{\widehat G,\widehat R_{q_R},\widehat S_{q_S},\theta}
  \quad\Longrightarrow\quad
  D_{\eta,{\rm loc}} .
\end{equation*}
Here \(G_\star\) is the growth score used by the localized envelope family.  The
log transformation and pilot fit produce a cone-comparable proxy score
\(\widehat G\).  The pilot quantiles \(\widehat R_{q_R}\) and
\(\widehat S_{q_S}\) then turn this proxy into a local deterministic denominator
\(D_{\eta,{\rm loc}}\).  This is the denominator tested under the
\emph{proxy-quantile} label in the main experiments.

The second part is the norm-polynomial tail-floor correction:
\begin{equation*}
  P_{\rm tail}(x)=1+\norm{x}^{\kappa_{\rm tail}}
  \quad\Longrightarrow\quad
  D_{\eta,{\rm final}}
  =
  \max\{D_{\eta,{\rm loc}},D_{\eta,{\rm tail}}\}.
\end{equation*}
It supplies the global Lyapunov input on all of \(\R^d\) without fresh
full-gradient evaluations in production.  In the reported diagnostics, the
tail-active fraction is small.  Thus the observed averages are driven mainly by
\(D_{\eta,{\rm loc}}\).

\begin{table}[H]
\centering
\caption{Symbols and meanings in the denominator design.}
\label{tab:design-roles}
\small
\begin{tabular}{lll}
\toprule
Object & Role & Meaning in the design \\
\midrule
\(\widehat G\) & proxy score & cheap switching score on the \(G_\star\)-scale \\
\(\widehat R_{q_R}\) & lower pilot quantile & boundary of the safe region \\
\(\widehat S_{q_S}\) & upper pilot quantile & onset of the local tail region \\
\(D_{\eta,{\rm loc}}\) & local denominator & denominator generated by the proxy-quantile envelope \\
\(P_{\rm tail}\) & tail indicator & cheap norm-polynomial score used by the tail floor \\
\(D_{\eta,{\rm tail}}\) & tail floor & controls the drift in the far tail \\
\(D_{\eta,{\rm final}}\) & final denominator & maximum of the local denominator and the tail floor \\
\bottomrule
\end{tabular}
\end{table}

\section{Envelope errors and stationary observable errors}
\label{sec:transition-stationary-errors}
\label{sec:transition-stationary}
Section~\ref{sec:tail-floor-correction} places the final proxy denominator in the stability setting required by the companion deterministic-envelope theory.  Once this is done, the Lyapunov and invariant-measure inputs are available.  It remains to control the error caused by replacing the \(G_\star\)-envelope with the proxy envelope.  This section shows how score, threshold, and response errors change the retention factor and then enter the stationary observable error bound.

\subsection{From level-set and threshold errors to envelope errors}
The next lemma converts score and threshold errors into localized-envelope
errors.

\begin{lemma}[Localized-envelope stability]
\label{lem:g-localized-envelope-stability}
Let
\begin{equation*}
  A_{G,R,S}(x)=(G(x)-R)_+^\theta+(G(x)-S)_+,
  \qquad 0<\theta\le1.
\end{equation*}
Then for any nonnegative scores \(G,\widetilde G\) and thresholds
\(R,S,\widetilde R,\widetilde S\),
\begin{equation}
\begin{aligned}
  |A_{\widetilde G,\widetilde R,\widetilde S}(x)-A_{G,R,S}(x)|
  &\le |\widetilde G(x)-G(x)|^\theta
       +|\widetilde R-R|^\theta \\
  &\quad +|\widetilde G(x)-G(x)|+|\widetilde S-S|.
\end{aligned}
\label{eq:g-localized-envelope-stability}
\end{equation}
\end{lemma}

\begin{proof}
For \(0<\theta\le1\), the map \(u\mapsto (u_+)^\theta\) is H\"older of
order \(\theta\):
\begin{equation*}
  |(a_+)^\theta-(b_+)^\theta|\le |a-b|^\theta .
\end{equation*}
This global H\"older form is used because it remains valid at activation
boundaries, where the soft term need not be Lipschitz when \(0<\theta<1\).
Apply it with
\[
  a=\widetilde G(x)-\widetilde R,
  \qquad
  b=G(x)-R.
\]
Since
\begin{equation*}
  |(\widetilde G-\widetilde R)-(G-R)|
  \le |\widetilde G-G|+|\widetilde R-R|,
\end{equation*}
and \((u+v)^\theta\le u^\theta+v^\theta\) for \(u,v\ge0\), the soft
term is bounded by the first two terms on the right-hand side of
\eqref{eq:g-localized-envelope-stability}.  The upper activation term follows
from the Lipschitz property of \(u\mapsto u_+\):
\begin{equation*}
  |(\widetilde G-\widetilde S)_+-(G-S)_+|
  \le |\widetilde G-G|+|\widetilde S-S|.
\end{equation*}
Combining the two estimates proves the result.
\end{proof}

The lemma becomes useful once the score and threshold errors are supplied by
the cone and quantile comparisons.  The following corollary records the resulting
localized-envelope error on a pilot region.

\begin{corollary}[Envelope error on a pilot region]
\label{cor:g-envelope-error-pilot-region}
Let \(\Omega_M\) be a pilot region on which \(G_\star(x)\le M\).  Suppose that
\[
  \widehat G(x)\asymp_{\delta,\tau}G_\star(x),
  \qquad x\in\Omega_M .
\]
Let \(R_{q_R}^\star\) and \(S_{q_S}^\star\) be the growth-score thresholds, and
let \(\widehat R_{q_R}\) and \(\widehat S_{q_S}\) be the fitted pilot thresholds
used by the proxy envelope.  Define
\[
  A_\star(x)=A_{G_\star,R_{q_R}^\star,S_{q_S}^\star}(x),
  \qquad
  \widehat A(x)=A_{\widehat G,\widehat R_{q_R},\widehat S_{q_S}}(x).
\]
Then, on \(\Omega_M\),
\begin{equation}
\begin{aligned}
  |\widehat A(x)-A_\star(x)|
  &\le
  \bigl[(e^\delta-1)(M+\tau)\bigr]^\theta
  + |\widehat R_{q_R}-R_{q_R}^\star|^\theta  \\
  &\quad
  +(e^\delta-1)(M+\tau)
  + |\widehat S_{q_S}-S_{q_S}^\star|.
\end{aligned}
\label{eq:g-envelope-error-pilot-region}
\end{equation}
Consequently, the population quantile sandwich
\eqref{eq:quantile-sandwich} and the empirical quantile error
\eqref{eq:empirical-quantile-error} give an explicit localized-envelope error
bound on \(\Omega_M\).
\end{corollary}

\begin{proof}
The cone comparison gives, on \(\Omega_M\),
\begin{equation*}
  |\widehat G(x)-G_\star(x)|
  \le (e^\delta-1)(G_\star(x)+\tau)
  \le (e^\delta-1)(M+\tau).
\end{equation*}
Applying Lemma~\ref{lem:g-localized-envelope-stability} with
\[
  G=G_\star,
  \qquad
  \widetilde G=\widehat G,
  \qquad
  R=R_{q_R}^\star,
  \qquad
  S=S_{q_S}^\star,
  \qquad
  \widetilde R=\widehat R_{q_R},
  \qquad
  \widetilde S=\widehat S_{q_S}
\]
gives \eqref{eq:g-envelope-error-pilot-region}.  The last statement follows
by inserting the threshold controls from \eqref{eq:quantile-sandwich} and
\eqref{eq:empirical-quantile-error}.
\end{proof}

\subsection{From envelope errors to modified-drift errors}
\label{sec:modified-drift-error}

The previous subsection controls the localized-envelope error.  This is not yet
an error of the Markov transition.  The chain uses an envelope through the
modified drift
\begin{equation*}
  b_A(x):=\Phi_A(x)b(x),
  \qquad
  \Phi_A(x)=\frac{1}{1+\eta^\alpha A(x)}.
\end{equation*}
Thus the envelope error must be converted into an error of the modified drift.
This conversion is needed because the stationary perturbation argument below
compares transition kernels.  The transition depends on \(A\) through the modified drift \(b_A\), not through
\(A\) directly.

The next lemma records this elementary step for two fixed deterministic
envelopes.  It will be applied with \(A=\widehat A\), the proxy-quantile
localized envelope, and \(B=A_\star\), the corresponding growth-score
envelope.

\begin{lemma}[Envelope error controls modified-drift error]
\label{lem:modified-drift-error}
Let \(A\) and \(B\) be two nonnegative deterministic envelopes.  Define
\begin{equation*}
  \Phi_A(x)=\frac{1}{1+\eta^\alpha A(x)},
  \qquad
  b_A(x)=\Phi_A(x)b(x),
\end{equation*}
and define \(\Phi_B\) and \(b_B\) in the same way.  Then
\begin{equation}
  |\Phi_A(x)-\Phi_B(x)|
  \le \eta^\alpha |A(x)-B(x)|.
  \label{eq:retention-lipschitz}
\end{equation}
Consequently,
\begin{equation}
  \|b_A(x)-b_B(x)\|
  \le
  \eta^\alpha \|b(x)\| |A(x)-B(x)|.
  \label{eq:weighted-envelope-distortion}
\end{equation}
The corresponding deterministic proposal-location residual satisfies
\begin{equation}
  \eta\|b_A(x)-b_B(x)\|
  \le
  \eta^{1+\alpha}\|b(x)\| |A(x)-B(x)|.
  \label{eq:euler-drift-perturbation}
\end{equation}
\end{lemma}

\begin{proof}
The scalar map \(a\mapsto(1+\eta^\alpha a)^{-1}\) has derivative
\begin{equation*}
  -\eta^\alpha(1+\eta^\alpha a)^{-2}.
\end{equation*}
Its absolute value is bounded by \(\eta^\alpha\) on \([0,\infty)\).  The
mean-value theorem gives \eqref{eq:retention-lipschitz}.  Multiplying by
\(\|b(x)\|\) gives \eqref{eq:weighted-envelope-distortion}.  Multiplying once
more by the Euler stepsize \(\eta\) gives
\eqref{eq:euler-drift-perturbation}.
\end{proof}

Thus the localized-envelope error from
Corollary~\ref{cor:g-envelope-error-pilot-region} enters the transition
through the modified drift.  The extra factor \(\eta\) appears when this
modified-drift error is converted into a one-step proposal-location residual.

\subsection{From denominator residuals to stationary observable errors}
\label{subsec:stationary-observable-residual}

 We now record how
these denominator residuals enter stationary observable errors.  The point is
not to reprove the Poisson bridge, but to use the companion stationary estimate
as a transfer rule from denominator-level errors to observable-level bias.

We use the stationary Poisson-bridge estimate from the companion
deterministic-envelope paper, namely \cite[Theorem~7.5]{companion}.  For an
observable \(H\) in the admissible test class, let \(u_H\) solve
\begin{equation*}
  \mathcal L u_H = H-\pi(H),
  \qquad \pi(u_H)=0 .
\end{equation*}
Let \(W_H\) denote the polynomial Poisson weight associated with the derivative
bounds on \(u_H\).  We assume throughout this subsection that the hypotheses
needed to invoke \cite[Theorem~7.5]{companion} hold for the final deterministic
denominator constructed above.

Set
\begin{equation*}
  A_{\rm final,\eta}(x)
  :=\frac{D_{\eta,{\rm final}}(x)-1}{\eta^\alpha},
  \qquad
  \Phi_{\rm final}(x)=\frac1{D_{\eta,{\rm final}}(x)} .
\end{equation*}
Since \(D_{\eta,{\rm final}}\ge1\), \(A_{\rm final,\eta}\) is a nonnegative
state-deterministic envelope.  Applying \cite[Theorem~7.5]{companion} gives
\begin{equation}
\begin{aligned}
  |\pi_\eta^{\rm final}(H)-\pi(H)|
  \le C_H\int W_H(x)
  \Big[
     \eta^{\gamma/2}
     +\eta^\alpha A_{\rm final,\eta}(x)\norm{b(x)}
     +\eta\,\mathsf{Var}_x(\varepsilon_{\eta,m})
  \Big]  \,\pi_\eta^{\rm final}(\dd x),
\end{aligned}
\label{eq:companion-final-residual-input}
\end{equation}
where, following the companion convention, the oracle contribution is measured
through the envelope-scaled perturbation
\begin{equation*}
  \varepsilon_{\eta,m}(x,U)
  :=
  -\,\frac{\zeta_m(x,U)}{1+\eta^\alpha A_\eta(x)},
  \qquad
  \E[\varepsilon_{\eta,m}(x,U)\mid x]=0,
\end{equation*}
and its conditional second moment
\begin{equation*}
  \mathsf{Var}_x(\varepsilon_{\eta,m})
  :=\E\big[\norm{\varepsilon_{\eta,m}(x,U)}^2\mid x\big].
\end{equation*}
This is the same scalar oracle-noise scale, with the minibatch index \(m\)
retained, as the quantity \(\mathsf{Var}_x(\varepsilon_\eta)\) used in
\cite[Theorem~7.5]{companion}.
In the exact-gradient case the last term is omitted.

We next decompose only the envelope residual in
\eqref{eq:companion-final-residual-input}.  Define the growth-score reference
envelope
\begin{equation*}
  A_\star(x)
  :=A_{G_\star,R_{q_R}^\star,S_{q_S}^\star}(x)
  =(G_\star(x)-R_{q_R}^\star)_+^\theta
   +(G_\star(x)-S_{q_S}^\star)_+,
\end{equation*}
and recall that
\begin{equation*}
  A_{\rm loc}(x)
  =A_{\widehat G,\widehat R_{q_R},\widehat S_{q_S}}(x).
\end{equation*}
Then
\begin{equation*}
  A_{\rm final,\eta}
  =
  A_\star
  +(A_{\rm loc}-A_\star)
  +(A_{\rm final,\eta}-A_{\rm loc}).
\end{equation*}
Therefore, pointwise,
\begin{equation}
\begin{aligned}
  \eta^\alpha A_{\rm final,\eta}(x)\norm{b(x)}
  \le {}&
  \eta^\alpha A_\star(x)\norm{b(x)}  \\
  &+\eta^\alpha |A_{\rm loc}(x)-A_\star(x)|\norm{b(x)} \\
  &+\eta^\alpha |A_{\rm final,\eta}(x)-A_{\rm loc}(x)|\norm{b(x)} .
\end{aligned}
\label{eq:denominator-residual-split}
\end{equation}
The three terms are, respectively, the reference-envelope residual, the
proxy-transfer residual, and the tail-floor residual.

The last term has a simpler denominator-level form.  Since
\(D_{\eta,{\rm loc}}=1+\eta^\alpha A_{\rm loc}\) and
\(D_{\eta,{\rm final}}\ge D_{\eta,{\rm loc}}\),
\begin{equation}
  \eta^\alpha |A_{\rm final,\eta}(x)-A_{\rm loc}(x)|
  =
  D_{\eta,{\rm final}}(x)-D_{\eta,{\rm loc}}(x).
  \label{eq:tail-floor-residual-denominator-form}
\end{equation}
It vanishes unless
\begin{equation*}
  D_{\eta,{\rm tail}}(x)>D_{\eta,{\rm loc}}(x),
\end{equation*}
that is, unless the polynomial tail floor is active.

Combining \eqref{eq:companion-final-residual-input},
\eqref{eq:denominator-residual-split}, and
\eqref{eq:tail-floor-residual-denominator-form} yields the following
bookkeeping form.

\begin{proposition}[Stationary residual for the proxy--tail denominator]
\label{prop:proxy-tail-stationary-residual}
Let \(D_{\eta,{\rm final}}\) be the calibrated final denominator, and let
\(\pi_\eta^{\rm final}\) be an invariant law of the corresponding chain.
Assume that \cite[Theorem~7.5]{companion} applies to the observable \(H\) with
Poisson weight \(W_H\).  Then
\begin{equation}
\begin{aligned}
  |\pi_\eta^{\rm final}(H)-\pi(H)|
  \le C_H\int W_H(x)
  \Big[
     &\eta^{\gamma/2}
     +\eta^\alpha A_\star(x)\norm{b(x)} \\
     &+\eta^\alpha |A_{\rm loc}(x)-A_\star(x)|\norm{b(x)} \\
     &+\bigl(D_{\eta,{\rm final}}(x)-D_{\eta,{\rm loc}}(x)\bigr)\norm{b(x)} \\
     &+\eta\,\mathsf{Var}_x(\varepsilon_{\eta,m})
  \Big] \,\pi_\eta^{\rm final}(\dd x).
\end{aligned}
\label{eq:proxy-tail-stationary-residual}
\end{equation}
In the exact-gradient case the term \(\eta\,\mathsf{Var}_x(\varepsilon_{\eta,m})\) is omitted.  If the
local proxy denominator is used without the tail floor, then
\(D_{\eta,{\rm final}}=D_{\eta,{\rm loc}}\), and the tail-floor residual is zero.
\end{proposition}

It remains to connect the proxy-transfer term with the calibration result.  On a
pilot region \(\Omega_M\), Corollary~\ref{cor:g-envelope-error-pilot-region}
gives
\begin{equation*}
  |A_{\rm loc}(x)-A_\star(x)|
  \le \mathcal E_{\rm env}(M),
  \qquad x\in\Omega_M,
\end{equation*}
where
\begin{equation}
\begin{aligned}
  \mathcal E_{\rm env}(M)
  := {}&
  \bigl[(e^\delta-1)(M+\tau)\bigr]^\theta
  + |\widehat R_{q_R}-R_{q_R}^\star|^\theta \\
  &+(e^\delta-1)(M+\tau)
  + |\widehat S_{q_S}-S_{q_S}^\star| .
\end{aligned}
\label{eq:env-error-short-hand}
\end{equation}
Therefore, on the pilot region,
\begin{equation}
\begin{aligned}
&\int_{\Omega_M} W_H(x)\eta^\alpha
  |A_{\rm loc}(x)-A_\star(x)|\norm{b(x)}
  \,\pi_\eta^{\rm final}(\dd x) \\
&\qquad\le
  \eta^\alpha \mathcal E_{\rm env}(M)
  \int_{\Omega_M} W_H(x)\norm{b(x)}
  \,\pi_\eta^{\rm final}(\dd x).
\end{aligned}
\label{eq:proxy-transfer-integrated-control}
\end{equation}

The calibration bound \eqref{eq:env-error-short-hand} is local on
\(\Omega_M\), whereas the proxy-transfer term in
\eqref{eq:proxy-tail-stationary-residual} is integrated over the full state
space.  We therefore split it into the pilot-region contribution and an
off-pilot remainder:
\begin{equation}
  \int W_H\,\eta^\alpha|A_{\rm loc}-A_\star|\norm{b}\,\dd\pi_\eta^{\rm final}
  =
  \int_{\Omega_M}(\cdots)\,\dd\pi_\eta^{\rm final}
  +R_{\rm tail}(M),
  \qquad
  R_{\rm tail}(M):=\int_{\Omega_M^c}(\cdots)\,\dd\pi_\eta^{\rm final}.
  \label{eq:proxy-transfer-split}
\end{equation}
The first term is bounded by \eqref{eq:proxy-transfer-integrated-control}.  For
the second term, the polynomial-tail domination condition
\eqref{eq:polynomial-tail-domination} and the final effective-linearity bound
\eqref{eq:final-effective-linearity} give a polynomial envelope for the
integrand on \(\Omega_M^c\).  Together with the uniform moment bound from
Proposition~\ref{prop:tail-corrected-stability}, this yields, for some
polynomial order \(m'\),
\begin{equation}
  R_{\rm tail}(M)
  \le
  C\int_{\Omega_M^c} W_H(x)(1+\norm{x}^{m'})\,\pi_\eta^{\rm final}(\dd x)
  \xrightarrow[M\to\infty]{}0,
  \label{eq:tail-remainder-vanishes}
\end{equation}
uniformly for \(\eta\le\eta_0\).  Thus the off-pilot proxy-transfer term is
controlled by the same tail mechanism that supplies the global stability input.

Collecting \eqref{eq:proxy-transfer-integrated-control},
\eqref{eq:proxy-transfer-split}, and \eqref{eq:tail-remainder-vanishes}, the
final stationary residual decomposes as
\begin{equation}
\begin{aligned}
  |\pi_\eta^{\rm final}(H)-\pi(H)|
 & \le C_H\Big[
    \underbrace{\eta^{\gamma/2}}_{\text{Euler}}
    +\underbrace{\textstyle\int W_H\,\eta\,\mathsf{Var}_x(\varepsilon_{\eta,m})\,\dd\pi_\eta^{\rm final}}_{\text{oracle covariance}}
    +\underbrace{\textstyle\int W_H\,\eta^\alpha A_\star\norm{b}\,\dd\pi_\eta^{\rm final}}_{G_\star\text{-envelope residual}} \\
    &+\underbrace{\eta^\alpha \mathcal E_{\rm env}(M)\textstyle\int_{\Omega_M} W_H\norm{b}\,\dd\pi_\eta^{\rm final}+R_{\rm tail}(M)}_{\text{proxy transfer}}
    +\underbrace{\textstyle\int(D_{\eta,{\rm final}}-D_{\eta,{\rm loc}})\norm{b}\,\dd\pi_\eta^{\rm final}}_{\text{tail floor}}
  \Big],
\end{aligned}
\label{eq:final-residual-bookkeeping}
\end{equation}
with the oracle-covariance term omitted in the exact-gradient case.  The
proxy-transfer contribution is the explicit pilot-region price
\(\mathcal E_{\rm env}(M)\) together with the remainder \(R_{\rm tail}(M)\),
which vanishes as \(M\to\infty\), and the tail-floor contribution is supported
only on the tail-activation region.

\section{Experiments}
\label{sec:experiments}

\paragraph{Experimental protocol.}
All experiments use the calibrated proxy-quantile denominator constructed in
Algorithm~\ref{alg:proxy-quantile-calibration}.  During sampling, the denominator
is a fixed state-deterministic function,
\begin{equation}
\label{eq:post-denominator}
  D_{\eta,\alpha,q_R,q_S,\theta}(x)
  =1+\eta^\alpha A_{\widehat G,q_R,q_S,\theta}(x),
  \qquad
  A_{\widehat G,q_R,q_S,\theta}(x)
  = (\widehat G(x)-\widehat R_{q_R})_+^\theta
    +(\widehat G(x)-\widehat S_{q_S})_+ .
\end{equation}
Here \(q_R,q_S\) set the activation thresholds, \(\alpha\) sets the denominator
multiplier, \(\theta\) sets the soft-ramp shape, and \(\eta\) is the transition
scale.  When global stability control is needed, this local denominator is
combined with the tail floor in \eqref{eq:tail-final-denominator}.

In the sampling update, the mini-batch randomness enters only through the
numerator.  Given an unbiased stochastic-gradient oracle
\(\widehat g_m(X_k,U_k)\), the transition is
\begin{equation}
\label{eq:experimental-sgld-update}
  X_{k+1}
  =
  X_k
  -
  \frac{\eta\,\widehat g_m(X_k,U_k)}
       {D_{\eta,\alpha,q_R,q_S,\theta}(X_k)}
  +
  \sqrt{2\beta^{-1}\eta}\,Z_{k+1},
  \qquad
  Z_{k+1}\sim N(0,I_d).
\end{equation}
Thus the denominator is not recomputed from the current mini-batch draw.

Throughout the experiments, \emph{proxy-quantile} denotes the deterministic
denominator produced from the fitted log-scale proxy score.  The
\emph{\(G_\star\)-envelope} applies the same quantile construction to the true
growth score and serves as a score reference.

We also use two standard baselines.  The random-denominator baseline puts the
current mini-batch gradient into the denominator:
\begin{equation}
\label{eq:random-denominator-baseline}
  D_{\eta,{\rm rand}}(x,U)
  =
  1+\eta^\alpha\bigl(1+\norm{\widehat g_m(x,U)}\bigr).
\end{equation}
Thus the same stochastic-gradient draw enters both the numerator and the
denominator.  The global-hard baseline is deterministic but nonlocalized:
\begin{equation}
\label{eq:global-hard-baseline}
  D_{\eta,{\rm gh}}(x)
  =
  1+\eta^\alpha\bigl(1+\norm{\nabla F_z(x)}\bigr).
\end{equation}
It uses the full-gradient scale everywhere and has no pilot-calibrated safe,
intermediate, or tail activation regions.

For the high-dimensional baseline sweep in
Experiment~\ref{sec:exp3-d50-comparison}, we also use scalar-rescaled versions
of these two baselines:
\begin{equation}
\label{eq:scaled-baseline-denominators}
  D_{\eta,{\rm rand}}^{(c)}(x,U)
  =
  1+\eta^\alpha c\bigl(1+\norm{\widehat g_m(x,U)}\bigr),
  \qquad
  D_{\eta,{\rm gh}}^{(c)}(x)
  =
  1+\eta^\alpha c\bigl(1+\norm{\nabla F_z(x)}\bigr).
\end{equation}
The scalar \(c>0\) changes only the overall taming strength of the two fixed
baseline forms; it is used as a small fairness sweep.  The proxy-quantile and \(G_\star\)-envelope
rows are kept fixed after quantile calibration.

In Experiment~\ref{sec:exp2-proxy-family}
we additionally report a global-polynomial proxy baseline.  This row is the
normalized, growth-score-scale analogue of the coarse polynomial proxy used in
Experiment~6 of the companion deterministic-envelope paper~\cite{companion}.

We report risk as a stationary observable; when an exact-gradient reference is
available, we report the absolute gap to that reference.

The experiments test the proxy-quantile construction in five steps.  We first check the proxy score and the resulting proxy envelope.  We then give the main high-dimensional quartic-regression comparison through observable gaps relative to exact-gradient and \(G_\star\)-envelope references.  The final two diagnostics isolate the far-tail safeguard and the production-stage computational cost.

For the quartic-regression experiments, the target risk is
\[
  F_z(w)=\frac1n\sum_{i=1}^n \frac14(a_i^\top w-y_i)^4
  +\frac{\lambda}{2}\norm{w}^2.
\]
We report \(F_z(W)\), \(\norm{W}^2\), and \(\norm{\nabla F_z(W)}\).  The exact-gradient chain provides the numerical reference scale.  The \(G_\star\)-envelope applies the same quantile construction to the ideal score \(G_\star\) and serves as the design reference.  The random-denominator and global-hard rows are the main competing baselines.  The global-polynomial row in Experiment~\ref{sec:exp2-proxy-family} links back to the polynomial proxy comparison in the companion experiments.

\subsection{Experiment 1: numerical proxy-sandwich diagnostic}
\label{sec:exp1-proxy-scale}

This diagnostic evaluates the fitted proxy score \(\widehat G\) on the quartic-regression pilot sample used below.  It reports the numerical level-set and positive-part sandwiches corresponding to
Proposition~\ref{prop:positive-part-sandwich}.

The log-shift is chosen from the pilot sample \({\mathcal P}\) by the fixed
preprocessing rule
\begin{equation}
\label{eq:exp-log-shift-rule}
  \tau
  =
  \max\Bigl\{10^{-6},\,
  0.01\,\operatorname{median}_{x\in {\mathcal P}}G_\star(x)
  \Bigr\}.
\end{equation}
The regression target is \(\log(G_\star+\tau)\), and the fitted proxy score is
reconstructed as
\begin{equation*}
  \widehat G(x)=[\exp\{\widehat h(x)\}-\tau]_+ .
\end{equation*}
The numerical floor prevents logarithmic instability near zero, while the factor
\(0.01\) keeps the shift small relative to the typical pilot growth-score scale.
In this diagnostic, \(\tau=4.97\times 10^{-3}\).

The parameter \(\delta\) below is not used by the sampling algorithm.  It is a
diagnostic tolerance for the empirical shifted log-scale error
\begin{equation*}
  \left|\log(\widehat G(x)+\tau)-\log(G_\star(x)+\tau)\right| .
\end{equation*}
Smaller \(\delta\) gives a tighter sandwich and hence a stricter test; the
reported violation masses are therefore expected to decrease as \(\delta\)
increases.

For the level-set diagnostic, let \(\widehat Q(q)\) be the empirical
\(q\)-quantile of \(\widehat G\).  Define the two growth-score thresholds
\begin{equation*}
  s_q^-(\delta)=e^{-\delta}\bigl(\widehat Q(q)+\tau\bigr)-\tau,
  \qquad
  s_q^+(\delta)=e^{\delta}\bigl(\widehat Q(q)+\tau\bigr)-\tau .
\end{equation*}
If the shifted log-scale error is at most \(\delta\), then the proxy sublevel set
satisfies
\begin{equation}
\label{eq:empirical-level-sandwich}
  \{G_\star\le s_q^-(\delta)\}
  \subset
  \{\widehat G\le \widehat Q(q)
  \}
  \subset
  \{G_\star\le s_q^+(\delta)\} .
\end{equation}
Table~\ref{tab:exp1-delta-levelset-sandwich} reports the empirical masses of the two violation sets on the diagnostic sample.  A miss is a point in \(\{G_\star\le s_q^-(\delta)\}\) that is excluded by the proxy sublevel set; a leak is a point included by the proxy sublevel set but lying outside \(\{G_\star\le s_q^+(\delta)\}\). The table reports the empirical masses of these two violation sets over the diagnostic sample.

\begin{table}[H]
\centering
\caption{Empirical level-set sandwich violations for the fitted proxy score
\(\widehat G\).}
\label{tab:exp1-delta-levelset-sandwich}
\begin{tabular}{ccccc}
\toprule
\(q\) & \(\delta\) & miss & leak & max viol. \\
\midrule
$0.70$ & $0.50$ & $0.032$ & $0.056$ & $0.056$ \\
$0.70$ & $0.75$ & $0.008$ & $0.018$ & $0.018$ \\
$0.70$ & $1.00$ & $0.000$ & $0.003$ & $0.003$ \\
\addlinespace
$0.90$ & $0.50$ & $0.008$ & $0.047$ & $0.047$ \\
$0.90$ & $0.75$ & $0.001$ & $0.013$ & $0.013$ \\
$0.90$ & $1.00$ & $0.000$ & $0.001$ & $0.001$ \\
\bottomrule
\end{tabular}
\end{table}

The same shifted log-scale tolerance also gives the positive-part comparison.
With \(s_q^-(\delta)\) and \(s_q^+(\delta)\) as above, the pointwise implication
is
\begin{equation}
\label{eq:empirical-positive-part-sandwich}
  e^{-\delta}\bigl(G_\star-s_q^+(\delta)\bigr)_+
  \le
  \bigl(\widehat G-\widehat Q(q)\bigr)_+
  \le
  e^{\delta}\bigl(G_\star-s_q^-(\delta)\bigr)_+ .
\end{equation}
Table~\ref{tab:exp1-delta-positive-part-sandwich} reports the empirical masses
where the two inequalities in \eqref{eq:empirical-positive-part-sandwich} fail.
A lower violation means that
\((\widehat G-\widehat Q(q))_+\) is smaller than the lower sandwich value
\(e^{-\delta}(G_\star-s_q^+(\delta))_+\).  An upper violation means that
\((\widehat G-\widehat Q(q))_+\) exceeds the upper sandwich value
\(e^\delta(G_\star-s_q^-(\delta))_+\).

\begin{table}[H]
\centering
\caption{Empirical positive-part sandwich violations for the fitted proxy score
\(\widehat G\).}
\label{tab:exp1-delta-positive-part-sandwich}
\begin{tabular}{ccccc}
\toprule
\(q\) & \(\delta\) & lower viol. & upper viol. & max viol. \\
\midrule
$0.70$ & $0.50$ & $0.088$ & $0.051$ & $0.088$ \\
$0.70$ & $0.75$ & $0.028$ & $0.014$ & $0.028$ \\
$0.70$ & $1.00$ & $0.006$ & $0.001$ & $0.006$ \\
\addlinespace
$0.90$ & $0.50$ & $0.053$ & $0.012$ & $0.053$ \\
$0.90$ & $0.75$ & $0.014$ & $0.001$ & $0.014$ \\
$0.90$ & $1.00$ & $0.002$ & $0.000$ & $0.002$ \\
\bottomrule
\end{tabular}
\end{table}

The violation masses decrease quickly as the diagnostic tolerance is relaxed.
Thus the fitted proxy score exhibits the two numerical sandwiches used by the
proxy theory: it places the relevant sublevel sets close to those of
\(G_\star\), and its positive parts above the activation thresholds are controlled
by the corresponding growth-score positive parts.  The next experiment checks
the resulting object at the denominator level by comparing the proxy-quantile
envelope with the \(G_\star\)-envelope.

\subsection{Experiment 2: proxy-quantile envelope versus \(G_\star\)-envelope}
\label{sec:exp2-proxy-family}
We next test whether the proxy score gives a useful denominator.  The proxy-quantile envelope applies the same quantile construction as the \(G_\star\)-envelope, but uses the proxy score \(\widehat G\) to define the activation regions.  The main question is whether these proxy-defined regions preserve the finite-chain behavior of the growth-score design.

Table~\ref{tab:exp2-proxy-family} compares this construction with three kinds of references.  The exact-gradient chain gives the numerical reference scale.  The \(G_\star\)-envelope gives the design reference.  The polynomial-proxy, random-denominator, and global-polynomial rows give competing low-cost or nonlocalized alternatives.  The last three columns report absolute gaps to the exact-gradient reference for the three observables.

For the global-polynomial row we use the normalized polynomial proxy
\begin{equation}
\label{eq:global-poly-baseline}
  G_{\rm poly}(w)
  =
  C_{\rm poly}\frac{1+\norm{w}^3}{1+\norm{w}},
  \qquad
  D_{\eta,{\rm gp}}(w)
  =
  1+\eta^\alpha\bigl(1+G_{\rm poly}(w)\bigr).
\end{equation}
This row is included to connect with Experiment~6 of the companion
paper~\cite{companion}, where the coarse polynomial envelope
\(C_{\rm poly}(1+\norm{w}^3)\) was used as a low-cost deterministic growth
indicator for quartic regression.  Because the present experiment works on the
normalized growth-score scale, we divide the polynomial envelope by
\(1+\norm{w}\).  The resulting denominator is deterministic and
nonlocalized; it does not use the fitted log-scale proxy or pilot-calibrated
activation geometry.

\begin{table}[H]
\centering
\caption{Log-scale and polynomial proxy envelopes.}
\label{tab:exp2-proxy-family}
\resizebox{\textwidth}{!}{%
\begin{tabular}{lrrrrrr}
\toprule
Method & Risk \(F_z(W)\) & Scale \(\norm{W}^2\) & Grad. size \(\norm{\nabla F_z(W)}\) & \(|\Delta F_z|\) & \(|\Delta \norm{W}^2|\) & \(|\Delta \norm{\nabla F_z}|\) \\
\midrule
exact-gradient reference & $1.370$ & $5.021$ & $1.801$ & $0.000$ & $0.000$ & $0.000$ \\
\(G_\star\)-envelope & $1.397$ & $5.076$ & $1.826$ & $0.050$ & $0.186$ & $0.050$ \\
proxy-quantile envelope & $1.385$ & $5.088$ & $1.811$ & $0.041$ & $0.192$ & $0.042$ \\
polynomial-proxy envelope & $1.864$ & $6.099$ & $2.213$ & $0.494$ & $1.078$ & $0.412$ \\
polynomial high-quantile envelope & $1.574$ & $5.489$ & $1.955$ & $0.206$ & $0.501$ & $0.163$ \\
random-denominator & $2.185$ & $6.378$ & $2.605$ & $0.815$ & $1.357$ & $0.804$ \\
global-polynomial denominator & $4.231$ & $9.635$ & $4.252$ & $2.861$ & $4.614$ & $2.451$ \\
\bottomrule
\end{tabular}}
\end{table}

Tables~\ref{tab:exp2-proxy-family} and~\ref{tab:exp2-gstar-relative-gaps} have different roles.  Table~\ref{tab:exp2-proxy-family} reports absolute finite-chain performance against the exact-gradient reference.  Table~\ref{tab:exp2-gstar-relative-gaps} removes this finite-chain reference scale and reports absolute gaps to the \(G_\star\)-envelope row.  This second table is the direct mechanism check for the local proxy layer: replacing the growth score by the proxy score changes the three observables by only about \(10^{-2}\), while the cruder baselines move them by one to two orders of magnitude more.

\begin{table}[H]
\centering
\caption{\(G_\star\)-relative observable gaps.}
\label{tab:exp2-gstar-relative-gaps}
\begin{tabular}{lccc}
\toprule
Method & \(|\Delta_{G_\star} F_z|\) & \(|\Delta_{G_\star}\norm{W}^2|\) & \(|\Delta_{G_\star}\norm{\nabla F_z}|\) \\
\midrule
proxy-quantile envelope & $0.012$ & $0.012$ & $0.015$ \\
polynomial-proxy envelope & $0.467$ & $1.023$ & $0.387$ \\
polynomial high-quantile envelope & $0.177$ & $0.413$ & $0.129$ \\
random-denominator & $0.788$ & $1.302$ & $0.779$ \\
global-polynomial denominator & $2.834$ & $4.559$ & $2.426$ \\
\bottomrule
\end{tabular}
\end{table}

Table~\ref{tab:exp2-proxy-family} shows that the proxy-quantile envelope nearly matches the \(G_\star\)-envelope.  Their observable values differ by about \(0.012\) in risk, \(0.013\) in squared norm, and \(0.015\) in gradient size.  Both rows are also close to the exact-gradient reference at the scale reported in the last three columns.

\subsection{Experiment 3: high-dimensional quartic-regression comparison}
\label{sec:exp3-d50-comparison}

We next test the proxy-quantile construction on a \(d=50\)
quartic-regression instance.  Unlike Experiment~1, this is a finite-chain
sampling comparison: the proxy score is fitted from a pilot sample, the
proxy-quantile denominator is fixed before production sampling, and the
resulting observables are compared with exact-gradient and tamed-denominator
baselines.  All rows use the same production stepsize \(\eta\), pilot protocol,
and production seeds.  The averages below are computed over twelve independent
production seeds.

\paragraph{Observable gaps to the exact-gradient reference.}
Table~\ref{tab:d50-observable-gaps-exact} reports signed observable gaps
relative to the exact-gradient chain at the same stepsize \(\eta\).  The
\(G_\star\)-envelope row applies the same quantile-localized response to the
growth score \(G_\star\), while the proxy-quantile row replaces \(G_\star\) by
the fitted proxy score \(\widehat G\).  The rows with \(c=0.5,1,2\) use the
scalar-rescaled baselines in \eqref{eq:scaled-baseline-denominators}.

\begin{table}[H]
\centering
\caption{Observable gaps relative to the exact-gradient chain.}
\label{tab:d50-observable-gaps-exact}
\begin{tabular}{lccc}
\toprule
Method & $\Delta F_z$ & $\Delta\norm{W}^2$ & $\Delta\norm{\nabla F_z(W)}$ \\
\midrule
\(G_\star\)-envelope & $+0.033\pm0.002$ & $+0.021\pm0.006$ & $+0.0076\pm0.0004$ \\
proxy-quantile & $+0.067\pm0.002$ & $-0.064\pm0.006$ & $+0.0163\pm0.0004$ \\
random denominator $c=0.5$ & $+0.461\pm0.010$ & $+0.679\pm0.025$ & $+0.105\pm0.002$ \\
global-hard denominator $c=0.5$ & $+0.305\pm0.007$ & $+0.461\pm0.017$ & $+0.070\pm0.001$ \\
random denominator $c=1$ & $+0.886\pm0.019$ & $+1.308\pm0.046$ & $+0.201\pm0.004$ \\
global-hard denominator $c=1$ & $+0.597\pm0.013$ & $+0.895\pm0.032$ & $+0.135\pm0.003$ \\
random denominator $c=2$ & $+1.637\pm0.035$ & $+2.452\pm0.082$ & $+0.368\pm0.007$ \\
global-hard denominator $c=2$ & $+1.139\pm0.024$ & $+1.720\pm0.059$ & $+0.257\pm0.005$ \\
\bottomrule
\end{tabular}
\end{table}

The proxy-quantile denominator remains close to the exact-gradient reference and to the \(G_\star\)-envelope row.  Its risk, squared-norm, and gradient-size gaps are all much smaller than those of the random-denominator and global-hard baselines under the tested scalar responses.  The scalar sweep over \(c\) is a small fairness check for these two fixed baseline forms: it changes only their overall taming strength, not their response geometry.  Even the least
aggressive tested baseline response, \(c=0.5\), remains well separated from the proxy-quantile and \(G_\star\)-envelope rows.  Thus the observed advantage is not masked by a reasonable rescaling of the baseline denominator.

\paragraph{Replacement cost relative to the \(G_\star\)-envelope.}
Table~\ref{tab:d50-observable-gaps-exact} uses the exact-gradient chain as an external sampling reference.  We now re-center the same observables at the \(G_\star\)-envelope.  This removes the exact-gradient reference scale and isolates the cost of replacing the growth score \(G_\star\) by the fitted proxy score \(\widehat G\). Table~\ref{tab:d50-gstar-relative-gaps} uses the \(G_\star\)-envelope as the internal reference.  This isolates the cost of replacing the growth score
\(G_\star\) by the fitted proxy score \(\widehat G\) after the same
quantile-localized response is applied.

\begin{table}[H]
\centering
\caption{Absolute observable gaps relative to the \(G_\star\)-envelope.}
\label{tab:d50-gstar-relative-gaps}
\begin{tabular}{lccc}
\toprule
Method & $|\Delta F_z|$ & $|\Delta\norm{W}^2|$ & $|\Delta\norm{\nabla F_z(W)}|$ \\
\midrule
proxy-quantile & $0.035\pm0.001$ & $0.084\pm0.002$ & $0.0088\pm0.0001$ \\
random denominator $c=0.5$ & $0.429\pm0.009$ & $0.658\pm0.022$ & $0.098\pm0.002$ \\
global-hard denominator $c=0.5$ & $0.272\pm0.006$ & $0.440\pm0.014$ & $0.062\pm0.001$ \\
random denominator $c=1$ & $0.853\pm0.018$ & $1.287\pm0.043$ & $0.194\pm0.003$ \\
global-hard denominator $c=1$ & $0.564\pm0.012$ & $0.874\pm0.029$ & $0.128\pm0.002$ \\
random denominator $c=2$ & $1.604\pm0.033$ & $2.431\pm0.080$ & $0.360\pm0.006$ \\
global-hard denominator $c=2$ & $1.106\pm0.023$ & $1.699\pm0.056$ & $0.249\pm0.004$ \\
\bottomrule
\end{tabular}
\end{table}

The \(G_\star\)-relative gaps are one order of magnitude smaller for the
proxy-quantile denominator than for the best random-denominator and global-hard
baselines in the table.  Thus the high-dimensional experiment supports the same
mechanism as Experiment~2: the fitted proxy score changes the
\(G_\star\)-envelope only mildly at the observable level, while cruder
denominator responses introduce substantially larger finite-chain distortion.

\subsection{Experiment 4: norm-polynomial tail-floor diagnostic}
\label{sec:exp4-tail-floor}

We next examine how often the norm-polynomial tail floor from
Section~\ref{sec:tail-floor-correction} changes the production-stage
denominator and what it does on the states where it is active.  To isolate this
question, we reuse the same $d=4$, $n=50$ quartic-regression instance and the
same proxy-calibration protocol as in Experiments~1--2.  The proxy score and
quantile thresholds are fitted once from the exact-gradient pilot trajectory
and then frozen throughout the diagnostic.

For this experiment we use the polynomial tail score
\[
  P_{\rm tail}(w)=1+\norm{w}^{2},
\]
with $T_{\rm tail}$ and $C_{\rm lin}$ given by the pilot $0.995$- and
$0.95$-quantiles, respectively, as in Section~\ref{sec:tail-floor-correction}.
The tail floor is said to be active at a recorded state $w$ when
\begin{equation}
  \mathcal A_{\rm tail}(w)
  :=
  \mathbf 1\!\left\{
    P_{\rm tail}(w)>T_{\rm tail},\quad
    \frac{P_{\rm tail}(w)}{C_{\rm lin}}>D_{\eta,{\rm loc}}(w)
  \right\}
  =1.
  \label{eq:experimental-tail-active}
\end{equation}
This is exactly the event on which the tail floor raises the local denominator;
the denominator construction itself is already given in
\eqref{eq:tail-final-denominator} and is not repeated here.

To measure the strength of the correction on those rare states, define
\[
  H_{\rm loc}(w)=\frac{G_\star(w)}{D_{\eta,{\rm loc}}(w)},
  \qquad
  H_{\rm final}(w)=\frac{G_\star(w)}{D_{\eta,{\rm final}}(w)}.
\]
We report the activation fraction
\begin{equation}
  p_{\rm tail}
  :=
  \Pr\!\left(\mathcal A_{\rm tail}(W)=1\right),
  \label{eq:tail-active-fraction}
\end{equation}
together with the conditional reduction
\begin{equation}
  \mathcal R_{\rm tail}
  :=
  1-
  \frac{
    \E[H_{\rm final}(W)\mid \mathcal A_{\rm tail}(W)=1]
  }{
    \E[H_{\rm loc}(W)\mid \mathcal A_{\rm tail}(W)=1]
  }.
  \label{eq:conditional-tail-reduction}
\end{equation}
The expectations in Table~\ref{tab:tail-activation-diagnostic} are empirical
averages pooled over the recorded active states from all seeds.

\begin{table}[H]
\centering
\caption{Norm-polynomial tail-floor diagnostic on the same quartic-regression target used in Experiments~1--2.}
\label{tab:tail-activation-diagnostic}
\small
\setlength{\tabcolsep}{6pt}
\begin{tabular}{lcc}
\toprule
$\eta$ & tail-active fraction $p_{\rm tail}$ & conditional reduction $\mathcal R_{\rm tail}$ \\
\midrule
$\eta=10^{-4}$ & $0.0002\pm0.0002$ & $31.4\%$ \\
$\eta=2\cdot10^{-4}$ & $0.0023\pm0.0016$ & $40.1\%$ \\
\bottomrule
\end{tabular}
\end{table}

The tail floor is active on only $0.018\%$ of recorded states at
$\eta=10^{-4}$ and about $0.23\%$ at $\eta=2\cdot10^{-4}$.  Yet on the states
where it does activate, the mean effective growth score is reduced by about
$31\%$ and $40\%$, respectively.  Thus the polynomial tail layer is genuinely
sparse but not negligible when it intervenes.  On this main
quartic-regression benchmark, the local proxy-quantile denominator determines
the dynamics on nearly all visited states, while the tail floor acts as a
selective safeguard on the small set of states where the local denominator is
insufficient.

\subsection{Experiment 5: computational cost diagnostic}
\label{sec:clock-overhead}

We close with a clock diagnostic in the \(d=50,n=1000\) quartic
stochastic-gradient setting.  The goal is to check whether the fitted proxy changes the production-stage cost regime.  We therefore report relative per-step time, normalized by the random-denominator baseline.

\begin{table}[H]
\centering
\caption{Production-stage clock cost.}
\label{tab:clock-overhead}
\small
\setlength{\tabcolsep}{5pt}
\begin{tabular}{lcc}
\toprule
Method & Per-step evaluation & rel. time \\
\midrule
random-denominator baseline. & SG oracle & $1.00$ \\
log-proxy envelope & SG oracle + fitted score & $1.09$ \\
exact-gradient ref. & full-gradient oracle & $3.90$ \\
\(G_\star\)-envelope & SG oracle + full-gradient score & $4.74$ \\
\bottomrule
\end{tabular}
\end{table}

The log-proxy envelope remains in the same production-stage cost regime as the standard stochastic-gradient denominators.  In this run it is about \(9\%\) slower than the random-denominator baseline, whereas the exact-gradient reference and the \(G_\star\)-envelope are about \(3.90\) and \(4.74\) times as
expensive, respectively.  

The one-time proxy construction cost was \(0.5445\)
seconds in total: \(0.4351\) seconds for the pilot trajectory, \(0.1076\) seconds for full-gradient score targets, and \(0.0018\) seconds for least squares and quantile calibration.  This corresponds to about \(897\) random-denominator production steps and is not incurred during production sampling.

\section{Discussion}

The proxy construction is not tied to one particular feature map.  The log-radial features used in the experiments are a low-cost choice for the reported examples, where the dominant growth is largely radial.  The analysis only uses the resulting comparison properties of the fitted score: level-set comparison, quantile transfer, and envelope-error control.  Other directional,
anisotropic, or model-specific features can be used if they provide the same inputs.

The norm-polynomial tail floor plays a limited role.  It restores the global growth control needed by the deterministic-envelope stability argument after the local denominator has been calibrated from pilot data.  It is not intended as an optimized tail-rate rule. Designing tail-rate-aware denominators would require choosing the tail score, activation level, and growth order together
with the drift growth and the finite-stepsize residual budget.

The present paper focuses on proxy construction, quantile calibration, and observable-level residual transfer for fixed deterministic envelopes.  Extending these results to full finite-step \(W_1\) or \(W_2\) convergence bounds requires additional stability and transport-dual Poisson estimates.  We therefore treat
this as a separate convergence problem rather than part of the calibration theory itself.

\appendix

\section{Reproducibility notes}
The accompanying experiment package separates the manuscript source, code, and data by experiment.  The tables in the main text are generated from the formal outputs for the proxy-sandwich diagnostics, proxy-family reference checks, the true-proxy \(d=50\) finite-chain comparison, the norm-polynomial tail-floor
diagnostic, and the clock-cost diagnostic.

Each experiment folder contains the scripts used to generate or merge the corresponding outputs, together with the reported CSV summaries.  The reported comparisons give fixed-protocol numerical evidence for the error channels studied in the paper.

\end{document}